\documentclass[11pt]{article}

\usepackage{fullpage}
\usepackage[T1]{fontenc}
\usepackage{gitinfo}
\usepackage[usenames,dvipsnames]{color}
\usepackage{stmaryrd}

\RequirePackage[colorlinks=true]{hyperref}
\hypersetup{
	linkcolor=[rgb]{0.3,0.3,0.6},
	citecolor=[rgb]{0.2, 0.6, 0.2},
	urlcolor=[rgb]{0.6, 0.2, 0.2}
}
\usepackage{mathpazo}
\usepackage{bm}
\usepackage{todonotes}
\usepackage{setspace}
\usepackage{scrextend}

\usepackage{rpmacros}
\usepackage{amsthm}
\usepackage{thmtools,thm-restate}

\numberwithin{equation}{section}
\declaretheoremstyle[bodyfont=\it,qed=\qedsymbol]{noproofstyle}

\declaretheorem[numberlike=equation]{observation}

\declaretheorem[name=Observation,numbered=no]{observation*}

\declaretheorem[numberlike=equation]{fact}

\declaretheorem[numberlike=equation]{theorem}

\declaretheorem[name=Theorem,numbered=no]{theorem*}

\declaretheorem[numberlike=equation]{lemma}
\declaretheorem[name=Lemma,numbered=no]{lemma*}

\declaretheorem[numberlike=equation]{corollary}
\declaretheorem[name=Corollary,numbered=no]{corollary*}

\declaretheorem[name=Proposition,numbered=no]{proposition*}

\declaretheorem[numberlike=equation]{claim}
\declaretheorem[name=Claim,numbered=no]{claim*}

\declaretheorem[name=Conjecture,numbered=no]{conjecture*}

\declaretheorem[name=Question,numbered=no]{question*}

\declaretheoremstyle[bodyfont=\it,qed=$\lozenge$]{defstyle} 

\declaretheorem[numberlike=equation,style=defstyle]{definition}
\declaretheorem[unnumbered,name=Definition,style=defstyle]{definition*}

\declaretheorem[unnumbered,name=Example,style=defstyle]{example*}

\declaretheorem[numberlike=equation,style=defstyle]{notation}
\declaretheorem[unnumbered,name=Notation=defstyle]{notation*}

\declaretheorem[unnumbered,name=Construction,style=defstyle]{construction*}

\declaretheorem[unnumbered,name=Remark,style=defstyle]{remark*}

\input{lazymacros} 
\usepackage{nth}
\usepackage{intcalc}
\usepackage{etoolbox}
\usepackage{xstring}

\usepackage{ifpdf}
\ifpdf
\else
\usepackage[quadpoints=false]{hypdvips}
\fi

\newcommand{\ECCCurl}[2]{http://eccc.hpi-web.de/report/\ifnumcomp{#1}{>}{93}{19}{20}#1/#2/}

\newcommand{\shortECCC}[2]{\texttt{\href{http://eccc.hpi-web.de/report/\ifnumcomp{#1}{>}{93}{19}{20}#1/#2/}{eccc:TR#1-#2}}}

\newcommand{\parseECCC}[1]{
\StrSubstitute{#1}{TR}{}[\tmpstring]%
\IfSubStr{\tmpstring}{/}{ 
\StrBefore{\tmpstring}{/}[\ecccyear]%
\StrBehind{\tmpstring}{/}[\ecccreport]%
}{
\StrBefore{\tmpstring}{-}[\ecccyear]%
\StrBehind{\tmpstring}{-}[\ecccreport]%
}%
\shortECCC{\ecccyear}{\ecccreport}}

\onehalfspace

\date{}
\title{Quadratic Lower Bounds for Algebraic Branching Programs and Formulas}

\author{Prerona Chatterjee\thanks{Tata Institute of Fundamental Research, Mumbai, India. Email: \texttt{prerona.chatterjee.tifr@gmail.com}. Research supported by the Department of Atomic Energy, Government of India, under project no. 12-R\&D-TFR-5.01-0500.}
\and Mrinal Kumar\thanks{Dept. of Computer Science \& Engineering, IIT Bombay, India. Email:
\texttt{mrinal@cse.iitb.ac.in}. A part of this work was done during a postdoctoral stay at University of Toronto.}
\and Adrian She\thanks{Dept. of Computer Science, University of Toronto. Email : \texttt{ashe@cs.toronto.edu}.}
\and Ben Lee Volk \thanks{Center for the Mathematics of Information, California Institute of Technology, USA. Email:\texttt{benleevolk@gmail.com}.}
}

\begin{document}
	
\maketitle

\begin{abstract}
We show that any Algebraic Branching Program (ABP) computing the polynomial $\sum_{i = 1}^n x_i^n$ has at least $\Omega(n^2)$ vertices. This improves upon the lower bound of $\Omega(n\log n)$, which follows from the classical result of Baur and Strassen \cite{Str73, BS83}, and extends the results in \cite{K19}, which showed a quadratic lower bound  for \emph{homogeneous} ABPs computing the same polynomial.

Our proof relies on a notion of depth reduction which is reminiscent of similar statements in the context of matrix rigidity, and shows that any small enough ABP computing the polynomial $\sum_{i=1}^n x_i^n$ can be depth reduced to essentially a homogeneous ABP of the same size which computes the polynomial $\sum_{i = 1}^n x_i^n + \varepsilon(\vx)$, for a structured ``error polynomial'' $\varepsilon(\vx)$. To complete the proof, we then observe that the lower bound in \cite{K19} is robust enough and continues to hold for all polynomials $\sum_{i = 1}^n x_i^n + \varepsilon(\vx)$, where  $\varepsilon(\vx)$ has the appropriate structure. 

We also use our ideas to show an $\Omega(n^2)$ lower bound of the size of algebraic formulas computing the elementary symmetric polynomial of degree $0.1n$ on $n$ variables. This is a slight improvement upon the prior best known formula lower bound (proved for a different polynomial) of $\Omega(n^2/\log n)$ \cite{Nec66, K85, SY10}. Interestingly, this lower bound is asymptotically better than $n^2/\log n$, the strongest lower bound that can be proved using previous methods. This lower bound also matches the upper bound, due to Ben-Or, who showed that elementary symmetric polynomials can be computed by algebraic formula (in fact depth-$3$ formula) of size $O(n^2)$. Prior to this work, Ben-Or's construction was known to be optimal only for algebraic formulas of depth-$3$ \cite{SW01}.   
\end{abstract}

\newpage

\section{Introduction}

Proving that there are explicit polynomials which are hard to compute is the template of many open problems in algebraic complexity theory. Various instances of this problem involve different definitions of explicitness, hardness and computation.

In the most general form, this is the well known $\VP$ vs. $\VNP$ question, which asks whether every ``explicit'' polynomial has a polynomial-size algebraic circuit. An algebraic circuit is a very natural (and the most general) algebraic computational model. Informally, it is a computational device which is given a set of indeterminates $\set{x_1, \ldots, x_n}$, and it can use additions and multiplications (as well as field scalars) to compute a polynomial $f \in \F[x_1, \ldots, x_n]$. The complexity of the circuit is then measured by the number of operations the circuit performs.

It is trivial to give an explicit $n$-variate polynomial which requires circuits of size $\Omega(n)$. It is also not hard to show that a degree-$d$ polynomial requires circuits of size $\Omega(\log d)$, since the degree can at most double in each operation. Thus, one trivially obtains a $\max\{n, \log d\} = \Omega(n + \log d)$ lower bound for an $n$-variate degree-$d$ polynomial.

A major result of Baur and Strassen \cite{Str73, BS83} gives an explicit $n$-variate degree-$d$ polynomial which requires circuits of size at least $\Omega(n \cdot \log d)$. On the one hand, this is quite impressive since  when $d=\poly(n)$, this gives lower bound which is super-linear in $n$. Such lower bounds for explicit functions in the analogous model of \emph{boolean} circuits are a long-standing and important open problem in boolean circuit complexity. On the other hand, this lower bound is  barely super-linear, whereas ideally one would hope to prove super-polynomial or even exponential lower bounds (indeed, it can be proved that ``most''  polynomials require circuits of size exponential in $n$).

Despite decades of work, this lower bound has not been improved, even though it has been reproved (using different techniques \cite{Smolensky97a, BenOr83}). Most of the works thus deal with restricted models of algebraic computation. For some, there exist exponential or at least super-polynomial lower bounds. For other, more powerful models, merely improved polynomial lower bound. We refer the reader to \cite{S15} for a comprehensive survey of lower bounds in algebraic complexity.


One such restricted model of computation for which we have better lower bounds is \emph{algebraic formulas}. Formulas are simply circuits whose underlying graph is a tree. Kalorkoti \cite{K85} has shown how to adapt Nechiporuk's method \cite{Nec66}, originally developed for boolean formulas, to prove an $\Omega(n^2/\log n)$ lower bound for an $n$-variate polynomial. \footnote{Since this exact statement appears to be folklore and is not precisely stated in Kalorkoti's paper, we give some more details about this result in \autoref{sec:previous}.}

\subsection{Algebraic Branching Programs}\label{sec:abp}
\emph{Algebraic Branching Programs} (ABPs, for short), defined below, are an intermediate model between algebraic formulas and algebraic circuits. To within polynomial factors, algebraic formulas can be simulated by ABPs, and ABPs can be simulated by circuits. It is believed that each of the reverse transformations requires a super-polynomial blow-up in the size (for some restricted models of computation, this is a known fact \cite{N91, R06, RY08, DMPY12, HY16}).

Polynomial families which can be efficiently computed by algebraic branching programs form the complexity class $\mathsf{VBP}$, and the determinant is a complete polynomial for this class under an appropriate notion of reductions. Thus, the famous Permanent vs.\ Determinant problem, unbeknownst to many, is in fact equivalent to showing super-polynomial lower bound for ABPs. In this paper, we focus on the question of proving lower bounds on the size of algebraic branching programs for explicit polynomial families. We start by formally defining an algebraic branching program. 

\begin{definition}[Algebraic Branching Programs]\label{def:abp}
	An Algebraic Branching Program (ABP) is a layered graph where each edge is labeled by an affine linear form and the first and the last layer have one vertex each, called the ``start'' and the ``end'' vertex respectively. 
		
  The polynomial computed by an ABP is equal to the sum of the weights of all paths from the start vertex to the end vertex in the ABP, where the weight of a path is equal to the product of the labels of all the edges on it. 
  
  The size of an ABP is the number of vertices in it. 
\end{definition}

While \autoref{def:abp} is quite standard, there are some small variants of it in the literature which we now discuss. These distinctions make no difference as far as super-polynomial lower bounds are concerned, since it can be easily seen that each variant can be simulated by the other to within polynomial factors, and thus the issues described here are usually left unaddressed. However, it seems that we are very far from proving super-polynomial lower bounds for general algebraic branching programs, and in this paper we focus on proving polynomial (yet still super-linear) lower bounds. In this setting, those issues do affect the results.

\paragraph*{Layered vs.\ Unlayered.} In \autoref{def:abp}, we have required the graph to be layered. We also consider in this paper ABPs whose underlying graphs are unlayered, which we call \emph{unlayered ABPs}. We are able to prove super-linear (but weaker) lower bounds for this model as well.

One motivation for considering layered graph as the ``standard'' model is given by the following interpretation. From the definition, it can be observed  that any polynomial computable by an ABP with $d$ layers and $\ell_i$ vertices in the $i$-th layer can be written as the (only) entry of the $1 \times 1$ matrix given by the product $M := \prod_{i = 1}^{d-1} M_i$, where $M_i$ is an $\ell_i \times \ell_{i+1}$ matrix with affine forms as entries. One natural complexity measure of such a representation is the total number of non-zero entries in those matrices, which is the number of edges in the ABP. Another natural measure, which can only be smaller, is the sums of dimensions of the matrices involved in the product, which is  the same as the number of vertices in the underlying graph.

Branching programs are also prevalent in boolean complexity theory, and in particular in the context of derandomizing the class $\mathsf{RL}$. In this setting again it only makes sense to talk about layered graphs.

Unlayered ABPs can also be thought of as (a slight generalization of) \emph{skew circuits}. These are circuits in which on every multiplication gate, at least one of the operands is a variable (or more generally, a linear function).

\paragraph*{Edge labels.} In \autoref{def:abp} we have allowed each edge to be labeled by an arbitrary affine linear form in the variables. This is again quite standard, perhaps inspired by Nisan's characterization of the ABP complexity of a non-commutative polynomial as the rank of an associated coefficients matrix \cite{N91}, which requires this freedom. A more restrictive definition would only allow each edge to be labeled by a linear function in 1 variable. On the other hand, an even more general definition, which we sometimes adopt, is to allow every edge to be labeled by an \emph{arbitrary} polynomial of degree at most $\Delta$. In this case we refer to the model as an ABP with edge labels of degree at most $\Delta$. Thus, the common case is $\Delta=1$, but our results are meaningful even when $\Delta = \omega(1)$. Note that this is quite a powerful model, which is allowed to use polynomials with super-polynomial standard circuit complexity ``for free''.

We will recall some of these distinctions in \autoref{sec:previous}, where we discuss previous results, some of which apply to several of the variants discussed here.

\subsection{Lower bounds for algebraic branching programs. }
Our first result is a quadratic lower bound on the size of any algebraic branching program computing some explicit polynomial.

\begin{restatable}{theorem}{mainABP}\label{thm:main ABP}
	Let $\F$ be a field and $n \in \N$ such that $\Char(\F) \nmid n$. Then any algebraic branching program over $\F$ computing the polynomial $\sum_{i = 1}^n x_i^n$ is of size at least $\Omega(n^2)$.
	
	When the ABP's edge labels are allowed to be polynomials of degree at most $\Delta$, our lower bound is $\Omega(n^2/\Delta)$.
\end{restatable}

For the unlayered case, we prove a weaker (but still superlinear) lower bound.

\begin{theorem}\label{thm:unlayered ABP}
	Let $\F$ be a field and $n \in \N$ such that $\Char(\F) \nmid n$. Then any unlayered algebraic branching program over $\F$ with edge labels of degree at most $\Delta$ computing the polynomial $\sum_{i = 1}^n x_i^n$ is of size at least $\Omega(n \log n / (\log \log n + \log \Delta))$.
\end{theorem}

\subsection{Lower bounds for algebraic formulas. } \label{sec:formulas}

Recall that an algebraic formula is an algebraic circuit where the underlying graph is a tree. For polynomials like $\sum_{i = 1}^n x_i^n$, a quadratic lower bound on the formula size is immediate. This is because the degree of the polynomial in every variable is $n$, and hence there must be at least $n$ leaves which are labeled by $x_i$ for every $i$. Therefore, for a formula lower bound to be non-trivial it should hold for a polynomial with bounded individual degrees (we discuss this further in \autoref{sec:previous}).

As our next result, we show an $\Omega(n^2)$ lower bound for a multilinear polynomial.

\begin{theorem}\label{thm:main formula}
	Let $n \in \N$ and let $\F$ be a field of characteristic greater than $0.1n$. Then any algebraic formula over $\F$ computing the elementary symmetric polynomial \[
	\esym(n, 0.1n) = \sum_{S \subseteq [n], \abs{S} = 0.1n}\prod_{j \in S} x_j \, , 
	\] is of size at least $\Omega(n^2)$.
\end{theorem}

As we describe in \autoref{sec:proof overview}, even though formulas can be simulated by ABPs with little overhead, and despite the fact that both proofs use similar ideas, this theorem is not an immediate corollary of \autoref{thm:main ABP}, but requires some work.

\subsection{Previous work}\label{sec:previous}

\paragraph*{ABP lower bounds. }
The best lower bound known for ABPs prior to this work is a lower bound of $\Omega(n\log n)$ on the number of edges for the same polynomial $\sum_{i = 1}^n x_i^n$. This follows from the classical lower bound of $\Omega(n\log n)$ by Baur and Strassen~\cite{Str73, BS83} on the number of multiplication gates in any algebraic circuit computing the polynomial $\sum_{i = 1}^n x_i^n$ and the observation that when converting an ABP to an algebraic circuit, the number of product gates in the resulting circuit is at most the number of edges in the ABP. \autoref{thm:main ABP} improves upon this bound quantitatively, and also qualitatively, since the lower bound is on the number of vertices in the ABP.

For the case of homogeneous ABPs,\footnote{An ABP is \emph{homogeneous} if the polynomial computed between the start vertex and any other vertex is a homogeneous polynomial. This condition is essentially equivalent to assuming that the number of layers in the ABP is upper bounded by the degree of the output polynomial.} a quadratic lower bound for the polynomial  $\sum_{i = 1}^n x_i^n$ was shown by Kumar \cite{K19}, and the proofs in this paper build on the ideas in \cite{K19}.  In a nutshell, the result in \cite{K19} is equivalent to a lower bound for ABPs computing the polynomial $\sum_{i = 1}^n x_i^n$ when the number of layers in the ABP is at most $n$. In this work, we generalize this to proving essentially the same lower bound as in \cite{K19}  for ABPs with an unbounded number of layers.

In general, an ABP computing an $n$-variate homogeneous polynomial of degree $\poly(n)$ can be homogenized with a polynomial blow-up in size. This is proved in a similar manner to the standard classical result which shows this statement for algebraic circuits \cite{Strassen1973b}. Thus, much like the discussion following \autoref{def:abp}, homogeneity is not an issue when one considers polynomial vs.\ super-polynomial sizes, but becomes relevant when proving polynomial lower bounds. In other contexts in algebraic complexity this distinction is even more sharp. For example, exponential lower bounds for homogeneous depth-$3$ circuits are well known and easy to prove \cite{NW97}, but strong enough exponential lower bounds for non-homogeneous depth-$3$ circuits would separate $\VP$ from $\VNP$ \cite{GKKS16}. 

For \emph{unlayered} ABPs, the situation is more complex. If the edge labels are only functions of one variable, then (for the same reasons discussed in \autoref{sec:formulas}) we need to only consider multilinear polynomials for the lower bound to be non-trivial. It is possible to adapt Nechiporuk's method \cite{Nec66} in order to obtain a lower bound of $\tilde{\Omega}(n^{3/2})$ for a multilinear polynomial in this model.
This is an argument attributed to Pudl\'{a}k and sketched by Karchmer and Wigderson \cite{KW93} for the boolean model of parity branching programs, but can be applied to the algebraic setting. However, this argument does not extend to the case where the edge labels are arbitrary linear or low-degree polynomials in the $n$ variables.
The crux of Nechiporuk's argument is to partition the variables into $m$ disjoint sets, to argue (using counting or dimension arguments) that the number of edges labeled by variables from each set must be somewhat large,\footnote{This is usually guaranteed by constructing a function or a polynomial with the property that given a fixed set $S$ in the partition, there are many subfunctions or subpolynomials on the variables of $S$ that can be obtained by different restrictions of the variables outside of $S$.} and then to sum the contributions over all $m$ sets.
This is hard to implement in models where a single edge can have a ``global'' access to all variables, since it is not clear how to avoid over-counting in this case.

As mentioned above, the lower bound of Baur-Strassen does hold in the unlayered case, assuming the edge labels are linear functions in the variables.
However, when we allow edge labels of degree at most $\Delta$ for some $\Delta \ge 2$, their technique does not seem to carry over.
Indeed, even if we equip the circuit with the ability to compute such low-degree polynomials ``for free'', a key step in the Baur-Strassen proof is the claim that if a polynomial $f$ has a circuit of size $\tau$, then there is a circuit of size $O(\tau)$ which computes all its first order partial derivatives, and this statement does not seem to hold in this new model.

It is possible to get an $\Omega(n \log n / \log \Delta)$ lower bound for this model, for a different polynomial, by suitably extending the techniques of Ben-Or \cite{BenOr83, BenOr94}.
Our lower bounds are weaker by at most a doubly-logarithmic factor; however, the techniques are completely different.
Ben-Or's proofs rely as a black-box on strong modern results in algebraic geometry, whereas our proofs are much more elementary.

\paragraph*{Formula lower bounds. } Before discussing prior results on formula lower bounds, we mention again that it is easy to see that any algebraic formula computing the polynomial $\sum_{i = 1}^n x_i^n$ has at least $n^2$ leaves, as since the degree of the polynomials in each of the variables $x_i$ is $n$, there must be at least $n$ leaves in the formula which are labeled by $x_i$. Thus, the number of leaves is at least $n\times n = n^2$. Therefore, for lower bounds for algebraic formulas to be interesting and non-trivial, they must be asymptotically larger than the sum of the individual degrees of the variables. One particularly interesting regime of parameters is when the target hard polynomial has constant individual degree (which can be taken to be $1$, essentially without loss of generality). 

The first non-trivial lower bound for algebraic formulas was shown by Kalorkoti \cite{K85} who proved a lower bound of $\Omega(n^{3/2})$ for an $n$-variate multilinear polynomial (in fact, for the $\sqrt{n} \times \sqrt{n}$ determinant). Kalorkoti's proof is essentially an algebraic version of  Nechiporuk's proof \cite{Nec66} of $\Omega(n^2/\log n)$ lower bound on the size of Boolean formula in the full binary basis (gates are allowed to be all Boolean functions on $2$ variables). Using very similar arguments, it is possible to show an $\Omega(n^2/\log n)$ lower bound for a different multilinear polynomial (the construction we describe here is based on a non-multilinear construction given in \cite{SY10}): let $\set{x_{i}^{(j)} : i \in [\log n], j \in [n/\log n]}$ and $\set{y_k : k \in [n]}$ be two sets of $n$ variables each, and fix a bijection $\iota : 2^{[\log n]} \to [n]$. The hard multilinear polynomial will be
\[
\sum_{j=1}^{n/\log n} \left( \sum_{S \subseteq [\log n]} y_{\iota(S)} \cdot \prod_{i \in S} x_i^{(j)} \right).
\]

It turns out that in general, this technique of Nechipurok and Kalorkoti \emph{cannot} prove a lower bound which is asymptotically better than $\left(n^2/\log n\right)$ (see Remark 6.18 in \cite{Jukna12}).  

In \autoref{thm:main formula}, we adapt the ideas used in \autoref{thm:main ABP} to prove an $\Omega(n^2)$ lower bound for an explicit $n$ variate multilinear polynomial. Thus, while the quantitative improvement over the previous lower bound is by a mere $\log n$ factor, it is interesting to note that \autoref{thm:main formula} gives an asymptotically better lower bound than the best bound that can be proved using the classical, well known techniques of Nechiporuk and Kalorkoti.

Our lower bound is for one of the elementary symmetric polynomials, which, due to a well known observation of Ben-Or, are also known to be computable by a formula (in fact, a depth-$3$ formula) of size $O(n^2)$ over all large enough fields. For a large regime of parameters, this construction was shown to be tight for depth-$3$ formulas (even depth-$3$ circuits) by Shpilka and Wigderson~\cite{SW01}. \autoref{thm:main formula} shows that elementary symmetric polynomials do not have formulas of size $o(n^2)$ regardless of their depth.

Prior to this work, the only super linear lower bound on the general formula complexity of elementary symmetric polynomials we are aware of is the $\Omega(n \log n)$ lower bound of Baur and Strassen \cite{BS83}, which as we mentioned earlier is in fact a lower bound on the stronger notion of \emph{circuit} complexity. As far as we understand, Kalorkoti's proof directly does not seem to give any better lower bounds on the formula complexity of these polynomials.

As another mildly interesting point, we note that Baur and Strassen \cite{BS83} also proved an \emph{upper bound} of $O( n \log n)$ on the circuit complexity of the elementary symmetric polynomials. Hence, our lower bound implies a super-linear separation between the formula complexity and the circuit complexity of an explicit polynomial family.

\subsection{Proof Overview}\label{sec:proof overview}
\paragraph*{ABP lower bounds. }
The first part in the proof of \autoref{thm:main ABP} is an extension of the lower bound proved in \cite{K19} for ABPs with at most $n$ layers. This straightforward but conceptually important adaptation shows that a similar lower bound holds for any polynomial of the form
\[
\sum_{i=1}^n x_i^n + \varepsilon(\vx)\,,
\]
where the suggestively named $\varepsilon(\vx)$ should be thought of as an ``error term'' which is ``negligible'' as far as the proof of \cite{K19} is concerned. The exact structure we require is that $\varepsilon(\vx)$ is of the form $\sum_{i=1}^r P_i Q_i + R$, where $P_i, Q_i$ are polynomials with no constant term and $\deg(R) \le n-1$. The parameter $r$ measures the ``size'' of the error, which we want to keep small, and the lower bound holds if, e.g., $r \le n/10$.

To argue about ABPs with $d$ layers, with $d > n$, we show that unless the size $\tau$ of the ABP is too large to begin with (in which case there is nothing to prove), it is possible to find a small set of vertices (of size about $\eta = \tau/d$) whose removal adds a small error term $\varepsilon (\vx)$ as above with at most $\eta$ summands, but also reduces the depth of the ABP by a constant factor. Repeatedly applying this operation $O(\log n)$ times eventually gives an ABP of depth at most $n$ while ensuring that we have not accumulated too much ``error'',\footnote{It takes some care in showing that the total number of error terms accumulated is at most $n/10$ as opposed to the obvious upper bound of $O(n\log n)$. In particular, we observe that the number of error terms can be upper bounded by a geometric progression with first term roughly $\tau/n$ and common ratio being a constant less than $1$.} so that we can apply the lower bound from the previous paragraph. 

In the full proof we have to be a bit more careful when arguing about the ABP along the steps of the proof above. The details are presented in \autoref{sec:layered}.

The proof of \autoref{thm:unlayered ABP} follows the same strategy, although the main impediment is that general undirected graphs can have much more complex structure then layered graphs. One of the main ingredients in our proof is (a small variant of) a famous lemma of Valiant \cite{Valiant77}, which shows that for every graph of depth $2^k$ with $m$ edges, it is possible to find a set of edges, of size at most $m/k$,  whose removal reduces the depth of the graph to $2^{k-1}$. This lemma helps us identify a small set of vertices which can reduce the depth of the graph by a constant factor while again accumulating small error terms.

Interestingly, Valiant originally proved this lemma in a different context, where he showed that linear algebraic circuits of depth $O(\log n)$ and size $O(n)$ can be reduced to a special type of depth-$2$ circuits (and thus strong lower bounds on such circuits imply super-linear lower bounds on circuits of depth $O(\log n)$). This lemma can be also used to show that \emph{boolean} circuits of depth $O(\log n)$ and size $O(n)$ can be converted to depth-$3$ circuits of size $2^{o(n)}$, and thus again strong lower bounds on depth-$3$ circuits will imply super-linear lower bounds on circuits of depth $O(\log n)$. Both of these questions continue to be well known open problems in algebraic and boolean complexity, and to the best of our knowledge, our proof is the first time Valiant's lemma is successfully used in order to prove circuit lower bounds for explicit functions or polynomials.

\paragraph*{Formula lower bounds. } The proof of \autoref{thm:main formula} follows along the lines of the proof of \autoref{thm:main ABP} but needs a few more ideas. A natural attempt at recovering formula lower bounds from ABP lower bounds is to convert the formula to an ABP and then invoke the lower bound in \autoref{thm:main ABP}. However, while it is easy to see that a formula can be transformed into an ABP with essentially no blow up in size, the ABP obtained at the end of this transformation is not necessarily layered. Therefore, we cannot directly invoke \autoref{thm:main ABP}. We can still invoke \autoref{thm:unlayered ABP}, but the lower bound thus obtained is barely super-linear and in particular weaker than the known lower bounds \cite{K85, SY10}. 

For the proof of \autoref{thm:main formula}, we apply the ideas in the proof of \autoref{thm:main ABP} in a non black-box manner. The proof is again in two steps, where in the first step, we show a lower bound of $\Omega(n^2)$ for polynomials of the form $\esym(n, 0.1n) + \varepsilon(\vx)$ for formulas of formal degree at most $0.1n$. Here, $\varepsilon(\vx)$ should again be thought of as an error term as in the outline of the proof of \autoref{thm:main ABP}. 

For formulas of formal degree greater than $0.1n$, we describe a simple structural procedure to reduce the formal degree of a formula, while maintaining its size. The cost of the procedure is that the complexity of error terms increases in the process. We argue that starting from a  formula of size at most $n^2$ (since otherwise, we are already done) computing $\esym(n, 0.1n)$, we can repeatedly apply this degree reduction procedure to obtain a formula of formal degree at most $0.1n$ which computes the polynomial $\esym(n, 0.1n) + \varepsilon(\vx)$, where the error terms is not too complex. Combining this with the robust lower bound for formulas of formal degree at most $0.1n$ completes the proof of \autoref{thm:main formula}.

\paragraph*{Elementary symmetric vs power symmetric polynomials. }
We remark that the lower bound in \autoref{thm:main ABP} also holds for elementary symmetric polynomials of degree $0.1n$ on $n$ variables. However, the proof seems a bit simpler and more instructive for the case of power symmetric polynomials and so we state and prove the theorem for these polynomials. As discussed in \autoref{sec:previous}, for the case of lower bounds for formulas, we are forced to work with multilinear polynomials and hence \autoref{thm:main formula} is stated and proved for elementary symmetric polynomials. In general, these lower bounds hold for any family of polynomials of high enough degree whose zeroes of multiplicity at least two lie in a low dimensional variety, or more formally, an analog of \autoref{clm:singular locus power sym} or \autoref{lem:esym singular locus robust} is true.   
\section{Notations and Preliminaries}
All logarithms in the paper are base 2.

We use some standard graph theory terminology: If $G$ is a directed graph and $(u,v)$ is an edge, $v$ is called the \emph{head} of the edge and $u$ the \emph{tail}. Our directed graphs are always acyclic with designated source vertex $s$ and sink vertex $t$. The \emph{depth} of a vertex $v$, denoted $\depth(v)$, is the length (in edges) of a longest path from $s$ to $v$. The depth of the graph, denoted by $\depth(G)$, is the depth of $t$.

For any two vertices $u$ and $v$ in an ABP, the polynomial computed between $u$ and $v$ is the sum of weights of all paths between $u$ and $v$  in the ABP. We denote this by $[u,v]$. 

The formal degree of a vertex $u$ in an ABP denoted $\fdeg(u)$, is defined inductively as follows: If $s$ is the start vertex of the ABP, $\fdeg(s)=0$. If $u$ is a vertex with incoming edges from $u_1, \ldots, u_k$, labeled by non-zero polynomials $\ell_1, \ldots, \ell_k$, respectively, then
\[
\fdeg(u) = \max_{i \in [k]} \set{ \deg(\ell_i) + \fdeg(u_i) }.
\]
It follows by induction that for every vertex $u$, $\deg ([s,u]) \le \fdeg(u)$ (however, cancellations can allow for arbitrary gaps between the two). The formal degree of the ABP is the maximal formal degree of any vertex in it.

We sometimes denote by $\vx$ the vector of variables $(x_1, \ldots, x_n)$, where $n$ is understood from the context. Similarly we use $\mathbf{0}$ to denote the $n$-dimensional vector $(0,0,\ldots,0)$.

\subsection{A decomposition lemma}
The following lemma gives a decomposition of a (possibly unlayered) ABP in terms of the intermediate polynomials it computes. Its proof closely resembles that of Lemma 3.5 of \cite{K19}. For completeness we prove it here for a slightly more general model.

\begin{lemma}[\cite{K19}]\label{lem:structure}
	Let $\mathcal{B}$ be a (possibly unlayered) algebraic branching program whose edge labels are arbitrary polynomials of degree at most $\Delta \le n/10$, which computes a degree $d$ polynomial $P \in \F[x_1, \ldots x_n]$, and has formal degree $d$. Set $d' = \floor{d/\Delta}$.
		
	 For any $i \in \set{1, 2, \ldots ,d'-1}$, let $S_i = \set{u_{i,1}, u_{i,2}, \ldots, u_{i,m}}$ be the set of all vertices in $\mathcal{B}$ whose formal degree is in the interval $[i \Delta, (i+1)\Delta)$.
	 
	 Then, there exist polynomials $Q_{i,1}, Q_{i,2}, \ldots, Q_{i,m}$ and $R_i$, each of degree at most $d-1$ such that
		\[P = \sum_{j=1}^{m} [s, u_{i,j}] \cdot Q_{i,j} + R_i\, .\]
\end{lemma}

\begin{proof}
Fix $i$ as above and set $S_i = \set{u_{i,1}, u_{i,2}, \ldots, u_{i,m}}$ as above (observe that since each edge label is of degree at most $\Delta$, $S_i$ is non empty). Further suppose, without loss of generality, that the elements of $S_i$ are ordered such that there is no directed path from $u_{i,j}$ to $u_{i,j'}$ for $j' > j$.

Consider the unlayered ABP $\mathcal{B}_1$ obtained from $\mathcal{B}$ by erasing all incoming edges to $u_{i,1}$, and multiplying all the labels of the outgoing edges from $u_{i,1}$ by a new variable $y_1$. The ABP $\mathcal{B}_1$ now computes a polynomial of the form
\[
	P'(y_1, x_1, \ldots, x_n) = y_1 \cdot Q_{i,1} + R_{i,1}
\]
where $P = P'([s,u_{i,1}], x_1, \ldots, x_n)$. $R_{i,1}$ is the polynomial obtained from $\mathcal{B}_1$ by setting $y_1$ to zero, or equivalently, removing $u_{i,1}$ and all its outgoing edges. We continue in the same manner with $u_{i,2}, \ldots, u_{i,m}$ to obtain
\[
P = \sum_{j=1}^m [s,u_{i,j}] \cdot Q_{i,j} + R_i.
\]
Indeed, observe that since there is no path from $u_{i,j}$ to $u_{i,j'}$ for $j' > j$, removing $u_{i,j}$ does not change $[s,u_{i,j'}]$. The bound on the degrees of $Q_{i,j}$ is immediate from the fact that the formal degree of the ABP is at most $d$ and $\fdeg(u_{i,j}) \ge 1$. It remains to argue the $\deg(R_i) \le d-1$.

The polynomial $R_i$ is obtained from $\mathcal{B}$ by erasing all the vertices in $S_i$ and the edges touching them. We will show that every path in the corresponding ABP computes a polynomial of degree at most $d-1$. Let $s=v_1, v_2, \ldots, v_r=t$  be such a path, which is also a path in $\mathcal{B}$. Let $v_k$ be the minimal vertex in the path whose degree (in $\mathcal{B}$) is at least $(i+1)\Delta$
(if no such $v_k$ exists, the proposition follows).
As $v_{k-1} \not \in S_i$, the formal degree of $v_{k-1}$ is at most $i \Delta - 1$.
The degree of the polynomial computed by this path is thus at most $i\Delta - 1 + \Delta + D = (i+1)\Delta - 1 + D$, where $D$ is the degree of product of the labels on the path $v_k, v_{k+1}, \ldots, t$. To complete the proof, it remains to be shown that $D \le d - (i+1)\Delta$.

Indeed, if $D \ge d - (i+1)\Delta + 1$ then since the degree of $v_k$ is at least $(i+1)\Delta$, there would be in $\mathcal{B}$ a path of formal degree at least $(i+1)\Delta +D \ge d + 1$, contradicting the assumption on $\mathcal{B}$.
\end{proof}

\section{A lower bound for Algebraic Branching Programs}	\label{sec:layered}
In this section we prove \autoref{thm:main ABP}. We start by restating it.

\mainABP*
	
For technical reasons, we work with a slightly more general model which we call \emph{multilayered} ABPs, which we now define. 
	
\begin{definition}[Multilayered ABP]
	Let $\mathcal{A}_1, \ldots, \mathcal{A}_k$ be $k$ ABPs with $d_1, \ldots, d_k$ layers and $\tau_1, \ldots, \tau_k$ vertices, respectively. A multilayered ABP $\mathcal{A}$, denoted by $\mathcal{A} = \sum_{i = 1}^k \mathcal{A}_i$, is the ABP obtained by placing $\mathcal{A}_1, \mathcal{A}_2, \ldots, \mathcal{A}_k$ in parallel and identifying their start and end vertices respectively. Thus, the polynomial computed by $\mathcal{A}$ is $\sum_{i=1}^k [\mathcal{A}_i]$, where $ [\mathcal{A}_i]$ is the polynomial computed by $\mathcal{A}_i$. 
	
The number of layers of $\mathcal{A}$ is $d := \max \set{d_1, \ldots, d_k}$. The size of $\mathcal{A}$ is the number of vertices in $\mathcal{A}$, and thus equals  \[\abs{ \mathcal{A}} := 2 + \sum_i (\tau_i-2) \, . \qedhere \] 
\end{definition}

This model is an intermediate model between (layered) ABPs and unlayered ABPs: given a multilayered ABP of size $\tau$ it is straightforward to construct an unlayered ABP of size $O(\tau)$ which computes the same polynomial.

\subsection{A robust lower bound for ABPs of formal degree at most $n$}
In this section, we prove a lower bound for the case where the formal degree of every vertex in the ABP is at most $n$. 
In fact, Kumar \cite{K19} has already proved a quadratic lower bound for this case.
	
\begin{theorem}[\cite{K19}]\label{thm:homog}
	Let $n \in \N$ and let $\F$ be a field such that $\Char(\F) \nmid n$. Then any algebraic branching program of formal degree at most $n$ which computes the polynomial $\sum_{i = 1}^n x_i^n$ has at least $\Omega(n^2)$ vertices. 
\end{theorem}
	
However, to prove \autoref{thm:main ABP}, we need the following more ``robust'' version of \autoref{thm:homog}, which gives a lower bound for a larger class of polynomials. For completeness, we also sketch an argument for the proof which is a minor variation of the proof of \autoref{thm:homog}.

\begin{theorem}\label{thm:robustHomog}
	Let $n \in \N$ and let $\F$ be field such that $\Char(\F) \nmid n$.
	Let $A_1(\vx), \ldots, A_r(\vx), B_1(\vx),  \ldots, B_r(\vx)$ and $R(\vx)$ be polynomials such that for every $i$, $A_i(\mathbf{0}) = B_i(\mathbf{0}) = 0$ and $R$ is a polynomial of degree at most $n-1$. Then, any algebraic branching program over $\F$, of formal degree at most $n$ and edge labels of degree at most $\Delta \le n/10$, which computes the polynomial 
	\[\sum_{i = 1}^n x_i^n + \sum_{j = 1}^r A_j\cdot B_j + R\]
	has at least $\frac{(n/2-r)n}{2\Delta}$ vertices. 
\end{theorem}

The proof of the theorem follows from \autoref{lem:structure} and the following lemma which is a slight generalization of  Lemma 3.1 in \cite{K19}. We include a proof for completeness.

\sloppy
\begin{lemma}\label{lem:varietyDimLB}
	Let $n \in \N$, and let $\F$ be an algebraically closed field such that $\Char(\F) \nmid n$. Let $\set{P_1, \ldots, P_m, Q_1, \ldots, Q_m, A_1, \ldots, A_r, B_1, \ldots, B_r}$ be a set of polynomials in $\F[x_1, \ldots, x_n]$ such that the set of their common zeros 
	\[
		V = \mathbb{V}(P_1, \ldots, P_m, Q_1, \ldots, Q_m, A_1, \ldots, A_r, B_1, \ldots, B_r) \subseteq \F^n
	\]
	is non-empty. Finally, suppose $R$ is a polynomial in $\F[\vx]$ of degree at most $n-1$, such that
	\[
		\sum_{i=1}^{n} x_i^n + \sum_{j=1}^{r} A_j \cdot B_j = R + \sum_{i=1}^{m} P_i \cdot Q_i.
	\]
	Then, $m \geq \frac{n}{2} - r$.
\end{lemma}
		
\begin{proof}
	Since $V \neq \emptyset$, $\dim(V) \geq n - 2m - 2r$ (see, e.g., Section 2.8 of \cite{Smi14}).  Thus, the set of zeros with multiplicity two of 
	\[
		\sum_{i=1}^{n} x_i^n - R = \sum_{i=1}^{m} P_i \cdot Q_i - \sum_{j=1}^{r} A_j \cdot B_j, 
	\]
	has dimension at least $n - 2m - 2r$. In other words, if $S$ is the set of common zeros of the set of all first order partial derivatives of $\sum_{i=1}^{n} x_i^n - R$, then $\dim(S) \geq n - 2m - 2r$. Up to scaling by $n$ (which is non-zero in $\F$, by assumption), the set of all first order partial derivatives of $\sum_{i=1}^{n} x_i^n - R$ is given by 
	\[
		\set{x_i^{n-1} - \frac{1}{n}\partial_{x_i}R}_{i \in [n]}.
	\]
	Thus, the statement of this lemma immediately follows from the following claim.
			
	\begin{claim}[Lemma 3.2 in \cite{K19}]\label{clm:singular locus power sym}
		Let $\F$ be an algebraically closed field, and $D$ a positive natural number. For every choice of polynomials $g_1, g_2, \ldots g_n \in \F[\vx]$ of degree at most $D-1$, the dimension of the variety 
		\[\mathbb{V}(x_1^D - g_1, x_2^D - g_2, \ldots, x_n^D - g_n)\] is zero.
	\end{claim}
	Indeed, the above claim shows that $0 = \dim(S) \geq n-2m-2r$, and so $m \geq \frac{n}{2} - r$. This completes the proof of \autoref{lem:varietyDimLB}.
\end{proof}

We now use \autoref{lem:structure} and \autoref{lem:varietyDimLB} to complete the proof of \autoref{thm:robustHomog}.

\begin{proof}[Proof of \autoref{thm:robustHomog}]
	Let $\mathcal{B}$ be an algebraic branching program of formal degree at most $n$, edge labels of degree at most $\Delta \le n/10$, and with start vertex $s$ and end vertex $t$, which computes
	\[\sum_{i = 1}^n x_i^n + \sum_{j = 1}^r A_j\cdot B_j + R.\]
	We may assume without loss of generality that $\F$ is algebraically closed, by interpreting $\mathcal{B}$ as an ABP over the algebraic closure of $\F$, if necessary.
		
	Let $n' = \lfloor n/\Delta \rfloor$, fix $k \in \set{1, 2, \ldots ,n'-1}$, and let $V_k = \set{v_{k,1}, v_{k,2}, \ldots, v_{k,m}}$ be the set of all vertices in $\mathcal{B}$ whose formal degree lies in the interval $[k\Delta, (k+1)\Delta)$.  Letting $P'_j = [s,v_{k,j}]$, by \autoref{lem:structure}, there exist polynomials $Q'_1, Q'_2, \ldots, Q'_m$ and $R'$, each of degree at most $n-1$ such that
	\[
		\sum_{i = 1}^n x_i^n + \sum_{j = 1}^r A_j\cdot B_j + R = \sum_{j=1}^{m} P'_j \cdot Q'_j + R'\, .
	\]
	Let $\alpha_j$, $\beta_j$ be the constant terms in $P'_j$, $Q'_j$ respectively. Then by defining
	\[P_j = P'_j - \alpha_j \qquad \text{and} \qquad Q_j = Q'_j - \beta_j,\]
	we have that
	\[
		\sum_{i=1}^{n} x_i^n + \sum_{j=1}^{r} A_j \cdot B_j = R'' + \sum_{j=1}^{m} P_j \cdot Q_j.
	\]
	Here, $R'' = -R + R' + \sum_{j=1}^{m} (\alpha_j \cdot Q'_j + \beta_j \cdot P'_j + \alpha_j \beta_j)$. We now have that for every $i$, the constant terms of $P_i$, $Q_i$ are zero and $\deg(R'') \leq n-1$.
	Let 
	\[
		\mathcal{V} = \mathbb{V}(P_1, \ldots, P_m, Q_1, \ldots, Q_m, A_1, \ldots, A_r, B_1, \ldots, B_r).
	\]
	Then $\mathbf{0} \in \mathcal{V}$, and so $\mathcal{V} \neq \emptyset$. Thus by \autoref{lem:varietyDimLB}, we know that $m \geq \frac{n}{2} - r$.
		
	Finally, for $ k\neq k' \in \set{1, 2, \ldots, n'-1}$, $V_k \cap V_{k'} = \emptyset$ if $d \neq d'$. Thus, the number of vertices in $\mathcal{B}$ must be at least
	\[
		\inparen{\dfrac{n}{2} - r} \cdot (n'-1) \geq \inparen{\frac{n}{2} - r}\cdot  \frac{n}{2\Delta}\,.
		\qedhere
	\]
\end{proof}

\subsection{A lower bound for the general case}

The following lemma shows how we can obtain, given an ABP with $d$ layers which computes a polynomial $F$, a multilayered ABP, whose number of layers is significantly smaller, which computes $F$ plus a small ``error term''.
	
\begin{lemma}\label{lem:shrinkage for a single layered abp}
	Let $\mathcal{A}$ be an ABP over a field $\F$ with $d$ layers, which computes the polynomial $F$ and has $m$ vertices. Let $s$ and $t$ be the start and end vertices of $\mathcal{A}$ respectively, and let $L = \set{u_1, u_2, \ldots, u_{|L|}}$ be the set of vertices in the $\ell$-th layer of $\mathcal{A}$. For every $i \in \set{1, 2, \ldots, |L|}$, let $\alpha_i$ and $\beta_i$ be the constant terms of $[s,u_i]$ and $[u_i,t]$ respectively. Furthermore, let $P_i$ and $Q_i$ be polynomials such that $[s,u_i] = P_i + \alpha_i$ and $[u_i,t] = Q_i + \beta_i$.
		
	Then, there is a multilayered ABP ${\mathcal{{A}}}'$,  with at most $\max\{\ell, d-\ell + 1\}$ layers and size at most $\abs{\mathcal{A}}$ that computes the polynomial 
	\[
		F - \sum_{i=1}^{\abs{L}} P_i \cdot Q_i + \sum_{i=1}^{\abs{L}} \alpha_i \cdot \beta_i .
	\]
\end{lemma}
	
\begin{proof}
	Let $u_1, u_2, \ldots, u_{|L|}$ be the vertices in $L$ as described, so that
	\[
		F = [s,t] = \sum_{i = 1}^{|L|} [s, u_{i}]\cdot [u_{i}, t] \, . 
	\]
	Further, for every $i \in \set{1, 2, \ldots, |L|}$, $[s,u_i] = P_i + \alpha_i$ and $[u_i,t] = Q_i + \beta_i$, where the constant terms of $P_i$ and $Q_i$ are zero (by definition). Having set up this notation, we can thus express the polynomial $F$ computed by $\mathcal{A}$ as 
	\[
		F = [s,t] = \sum_{i = 1}^{|L|} (P_i + \alpha_i)\cdot (Q_i + \beta_i) \, .
	\]
	On further rearrangement, this gives 
	\[
		F - \inparen{\sum_{i = 1}^{|L|} P_i \cdot Q_i} + \inparen{\sum_{i = 1}^{|L|} \alpha_i\cdot \beta_i} = \inparen{\sum_{i = 1}^{|L|} \alpha_i \cdot (Q_i + \beta_i)} + \inparen{\sum_{i = 1}^{|L|} (P_i + \alpha_i)\cdot \beta_i} \, .
	\] 
	This is equivalent to the following expression.
	\[
		F - \inparen{\sum_{i = 1}^{|L|} P_i \cdot Q_i} + \inparen{\sum_{i = 1}^{|L|} \alpha_i\cdot \beta_i} =  \inparen{\sum_{i = 1}^{|L|} \alpha_i\cdot [u_i, t]} + \inparen{\sum_{i = 1}^{|L|} [s, u_i]\cdot \beta_i} \, .
	\] 
	Now, observe that the polynomial $\sum_{i=1}^{|L|} [s, u_i] \cdot \beta_i$ is computable by an ABP $\mathcal{B}$ with $\ell+1$ layers, obtained by just keeping the vertices and edges within first $\ell$ layers of $\mathcal{A}$ and the end vertex $t$, deleting all other vertices and edges, and connecting the vertex $u_i$ in the $\ell$-th layer to $t$ by an edge of weight $\beta_i$.  
	Similarly, the polynomial $\sum_{i=1}^{|L|} \alpha_i\cdot [u_i, t]$ is computable by an ABP $\mathcal{C}$ with at most $(d-\ell + 1)+1$ layers,  whose set of vertices is $s$ along the vertices in the layers $\ell, \ell + 1, \ell+2, \ldots, d$ of $\mathcal {A}$. From the definition of $\mathcal{B}$ and $\mathcal{C}$, it follows that the multilayered ABP $\tilde{\mathcal{A}}$ obtained by taking the sum of $\mathcal{B}$ and $\mathcal{C}$ has at most $\max\set{\ell+1, d-\ell+2}$ layers. 
		
	We are almost done with the proof of the lemma, except for the upper bound on the number of vertices of the resulting multilayered ABP $\tilde{\mathcal{A}}$, and the fact that the upper bound on the depth is slightly weaker than claimed. Both these issues can be solved simultaneously.
	
	 The vertices in $L$ appear in both the ABP $\mathcal{B}$ and the ABP $\mathcal{C}$ and are counted twice in the size of $\tilde{\mathcal{A}}$. However, every other vertex is counted exactly once.  Hence, 
	\begin{equation}\label{eqn:overlap}
		\abs{\mathcal{B}} + \abs{\mathcal{C}} =  \abs{\mathcal{A}} + \abs{L} \, .
	\end{equation}
 	In order to fix this issue, we first observe that the edges between the vertices in the $\ell$-th layer of $\mathcal{B}$ and the end vertex $t$ are labeled by $\beta_{1}, \beta_2, \ldots, \beta_{|L|}$, all of which are field constants. In the following claim, we argue that for ABPs with this additional structure, the last layer is redundant and can be removed.

	\begin{claim}\label{clm:last layer removal}
		Let ${\mathcal{M}}$ be an ABP over $\F$ with $k+1$ layers and edge labels of degree at most $\Delta$ such that the labels of all the edges between the $k$-th layer of ${\cal M}$ and its end vertex are scalars in $\F$. Then, there is an ABP $\mathcal{M}'$ with  $k$ layers computing the same polynomial as $\mathcal{M}$, with edge labels of degree at most $\Delta$, such that
		\[\abs{\mathcal{M}'} \leq \abs{\mathcal{M}} - \abs{V} \, , \]
		where $V$ is the set of vertices in the $k$-th layer of $\mathcal{M}$.
	\end{claim}
	
	An analogous statement, with an identical proof, is true if we assume that all edge labels between the first and second layer are scalars in $\F$.

	We first use \autoref{clm:last layer removal} to complete the proof of the lemma. As observed above, the edge labels between the last layer $L$ of $\mathcal{B}$ and its end vertex are all constants. Hence, by \autoref{clm:last layer removal}, there is an ABP $\mathcal{B}'$ which computes the same polynomial as $\mathcal{B}$ such that $\abs{\mathcal{B}'} \leq \abs{\mathcal{B}} - \abs{L}$, and $\mathcal{B}'$ has only $\ell$ layers. Similarly, we can obtain an ABP $\mathcal{C}'$ with at most $d-\ell+1$ layers.
	
	We consider the multilayered ABP $\mathcal{A}'$ by taking the sum of $\mathcal{B}'$ and $\mathcal{C}'$. Clearly, the number of layers in $\mathcal{A}'$ is at most $\max\{\ell, d-\ell+1\}$ and the size is at most 
	\[\abs{\mathcal{A}'} \leq  \abs{\mathcal{B}'} + \abs{\mathcal{C'}} \leq  \inparen{\abs{\mathcal{B}} - \abs{L}} + \inparen{\abs{\mathcal{C} }- \abs{L}} \leq \abs{\mathcal{A}} \, .\]
	Here, the second inequality follows by \autoref{clm:last layer removal} and the last one follows by \autoref{eqn:overlap}. To complete the proof of the lemma, we now prove \autoref{clm:last layer removal}. 
\end{proof}
	
\begin{proof}[Proof of \autoref{clm:last layer removal}]
	For the proof of the claim, we  focus on the $k$-th and $(k-1)$-st layer of $\mathcal{M}$. To this end, we first set up some notation. Let $\set{v_1, v_2, \ldots, v_r}$ be the set of vertices in the $k$-th layer of $\mathcal{M}$, $\set{u_1, u_2, \ldots, u_{r'}}$ be the set of vertices in $(k-1)$-st layer of $\mathcal{M}$, and $a$, $b$ denote the start and the end vertices of $\mathcal{M}$ respectively. Then, the polynomial computed by $\mathcal{M}$, can be decomposed as 
	\[
		[a,b] = \sum_{i = 1}^r [a, v_i] \cdot  [v_i, b]\, .
	\]
	Note that $(v_i,b)$ is an edge in the ABP. Similarly, the polynomial $[a,v_i]$ can be written as 
	\[
		[a,v_i] = \sum_{j = 1}^{r'} [a, u_j] \cdot [u_j, v_i] \, .
	\]
	Combining the two expressions together, we get  
	\[
		[a,b] = \sum_{i = 1}^r [v_i, b] \cdot  \inparen{\sum_{j = 1}^{r'} [u_j,v_i] \cdot [a, u_j]}\, ,
	\]
	which on further rearrangement,  gives us 
	\begin{equation} \label{eqn:decomposition using second last layer}
		[a,b] = \sum_{j = 1}^{r'} \inparen{\sum_{i = 1}^{r} [v_i,b] [u_j, v_i]} \cdot [a, u_j] \, .
	\end{equation}
 	From the hypothesis of the claim, we know that for every $i \in [r]$, the edge label $[v_i,b]$ is a field constant, and the edge label $[u_j, v_i]$ is a polynomial of degree at most $\Delta$. Thus, for every $j \in [r']$, the expression $\inparen{\sum_{i = 1}^{r} [u_i,b] [u_j,v_i]}$ is a polynomial of degree at most $\Delta$.
 	
 	This gives us the following natural construction for the ABP $\mathcal{M}'$ from $\mathcal{M}$. We delete the vertices $v_1, v_2, \ldots, v_r$ in $\mathcal{M}$ (and hence, all edges incident to them), and for every $j \in \set{1, 2, \ldots, r'}$, we connect the vertex $u_j$ with the end vertex $b$ using an edge with label $\inparen{\sum_{i = 1}^{r} [v_i,b] [u_j, v_i]}$. The upper bound on the size and the number of layers of $\mathcal{M}'$ is immediate from the construction, and that it computes the same polynomial as $\mathcal{M}$ follows from \autoref{eqn:decomposition using second last layer}.
 \end{proof}	

We now state and prove a simple generalization of \autoref{lem:shrinkage for a single layered abp} for a multilayered ABP. 
	
\begin{lemma}\label{lem:shrinkage general}
	Let $\mathcal{A} = \sum_{i=1}^{m} \mathcal{A}_i$ be a multilayered ABP with $d$ layers over a field $\F$ computing the polynomial $F$, such that each $\mathcal{A}_i$ is an ABP with $d_i$ layers.
	Also, let $\ell_{i,j}$ be the number of vertices in the $j$-th layer of $\mathcal{A}_i$ ($\ell_{i,j} = 0$ if $\mathcal{A}_i$ has fewer than $j$ layers), and $\ell = \min_{j \in (d/3, 2d/3)} \set{\sum_{i=1}^m \ell_{i,j}}$.
	
	Then, there is a multilayered ABP  with at most $2d/3$ layers and size at most $\abs{\mathcal{A}}$ that computes a polynomial of the form
	\[F - \sum_{i=1}^{\ell} P_i \cdot Q_i + \delta \, , \]
	where $\set{P_1, \ldots, P_{\ell}, Q_1, \ldots, Q_{\ell}}$ is a set of non-constant polynomials with constant term zero and $\delta \in \F$.
\end{lemma}
	
\begin{proof}
	Let $j_0 \in (d/3, 2d/3)$ be the natural number which minimizes the quantity $\sum_{i=1}^m \ell_{i,j}$, and let $S \subseteq [m]$ be the set of all indices $i$ such that $\mathcal{A}_i$ has at least $j_0$ layers. Let $\mathcal{A}' = \sum_{i \in S}\mathcal{A}_i$ and $\mathcal{A}'' = \sum_{i \notin S}\mathcal{A}_i$. Thus, 
	\[
		\mathcal{A} = \mathcal{A}' + \mathcal{A}''.
	\]		
	Here, $\mathcal{A}'' = \sum_{i \notin S}\mathcal{A}_i$ is a multilayered ABP with at most $2d/3$ layers. Moreover, $\abs{\mathcal{A}} = \abs{\mathcal{A}'} + \abs{\mathcal{A}''}$.
		
	To complete the proof of this lemma, we will now apply~\autoref{lem:shrinkage for a single layered abp} to every ABP in $\mathcal{A}'$. For every $i \in S$, we know that there exist some polynomials $P_{i,1}, \ldots, P_{i, \ell_{i, {j_0}}}, Q_{i,1}, \ldots, Q_{i, \ell_{i, {j_0}}}$ with constant terms zero and a constant $\delta_i$, such that  
	\[
		F_i - \sum_{r = 0}^{\ell_{i, {j_0}}} P_{i, r}Q_{i,r} + \delta_i
	\]
	can be computed by a multilayered ABP. Let us denote this multilayered ABP by $\mathcal{B}_i$. From \autoref{lem:shrinkage for a single layered abp}, we know that $\mathcal{B}_i$ has at most $\max\{j_0, d_i-j_0+1\} \leq 2d/3$ layers and size at most $\abs{\mathcal{A}_i}$. Taking a sum over all $i \in S$ and re-indexing the summands, we get that there exist polynomials $P_{1}, \ldots, P_{\ell}, Q_{1}, \ldots, Q_{\ell}$ with constant terms zero and a constant $\delta$ such that the polynomial 
	\[
		\sum_{i\in S} F_i - \sum_{r = 0}^{\ell} P_{r}Q_{r} + \delta
	\]
	is computable by a multilayered ABP $\mathcal{B} = \sum_{i \in S} \mathcal{B}_i$ with at most $2d/3$ layers and size at most $\sum_{i \in S} \abs{\mathcal{A}_i} \leq \abs{\mathcal{A}'}$. Now, by combining the multilayered ABPs $\mathcal{B}$ and $\mathcal{A}''$, we get that the polynomial 
	\[
		F  - \sum_{r = 0}^{\ell} P_{r}Q_{r} + \delta
	\]
	is computable by a multilayered ABP with at most $2d/3$ layers and size at most $\abs{\mathcal{A}}$.
\end{proof}
	
We now use \autoref{lem:shrinkage general} to prove \autoref{thm:main ABP}. 
	
\begin{proof}[Proof of \autoref{thm:main ABP}]
	Let $\mathcal{A}$ be a multilayered ABP with $d_0$ layers which computes the polynomial $\sum_{i = 1}^n x_i^n$. As before we may assume without loss of generality that the underlying field $\F$ is algebraically closed. Note that if $d_0$ is at most $n/\Delta$, then by \autoref{thm:robustHomog}, we know that $\abs{\mathcal{A}}$ is at least $\Omega(n^2)$ and we are done. Also, if $d_0 > n^2/\Delta$, then again we have our lower bound since each layer of $\mathcal{A}$ must have at least one vertex. Thus, we can assume that $n/\Delta \le d_0 \le n^2/\Delta$. 
		
	The proof idea is to iteratively make changes to $\mathcal{A}$ till we get a multilayered ABP $\mathcal{A}'$ of formal degree at most $n$ that computes a polynomial of the type 
	\[\sum_{i = 1}^n x_i^n + \sum_{j = 1}^r A_j\cdot B_j + R\]
	where $r \leq n/10$ and $A_1(\vx), \ldots, A_r(\vx), B_1(\vx), \ldots, B_r(\vx), R(\vx)$ are polynomials such that for every $i$, $A_i(\mathbf{0}) = B_i(\mathbf{0}) = 0$ and $R$ has degree at most $n-1$. 
	Once we have this, we can invoke \autoref{thm:robustHomog} and get the required lower bound.
		
	We now explain how to iteratively obtain $\mathcal{A}'$ from $\mathcal{A}$. In one step, we ensure the following.
		
	\begin{claim}\label{clm:oneStep}
		Let $\mathcal{A}_k$ be a multilayered ABP with edge labels of degree at most $\Delta$, $d_k \ge n/\Delta$ layers and size at most $\tau$ that computes a polynomial of the form $\sum_{i = 1}^n x_i^n + \sum_{j = 1}^r A_j\cdot B_j + R$ where $A_1(\vx), \ldots, A_r(\vx), B_1(\vx), \ldots, B_r(\vx), R(\vx)$ are polynomials such that for every $j$, $A_j(\mathbf{0}) = B_j(\mathbf{0}) = 0$ and $R$ has degree at most $n-1$.
			
		If $\tau \leq 0.001 n^2/\Delta$, then there exists a multilayered ABP $\mathcal{A}_{k+1}$ with at most $2 d_k/3$ layers and size at most $\tau$ which computes a polynomial of the form 
		\[\sum_{i = 1}^n x_i^n + \sum_{j = 1}^{r'} A'_j\cdot B'_j + R'\, ,\]
		such that $r' \leq  r + 0.005\frac{n^2}{\Delta \cdot d_k}$ and $A'_1(\vx), \ldots, A'_{r'}(\vx), B'_1(\vx), \ldots, B'_{r'}(\vx), R'(\vx)$ are polynomials such that for every $i$, $A'_i(\mathbf{0}) = B'_i(\mathbf{0}) = 0$ and $R'$ has degree at most $n-1$.
	\end{claim}
	
	Before moving on to the proof of  \autoref{clm:oneStep}, we first use it to complete the proof of  \autoref{thm:main ABP}. Let us set $\mathcal{A}_0 = \mathcal{A}$. Then, $\mathcal{A}_0$ is a multilayered ABP with $d_0$ layers and size at most $\tau $ that computes the polynomial $\sum_{i = 1}^n x_i^n$.
		
	If $\tau \geq 0.001 n^2/\Delta$, the statement of the theorem follows. Otherwise, we apply \autoref{clm:oneStep} iteratively $K$ times, as long as the number of layers is more than $n/\Delta$, to  eventually get a multilayered ABP $\mathcal{A}' = \mathcal{A}_{K}$ with $d' \le n/\Delta$ layers. Let $d_0, \ldots, d_{K-1},d_K$ denote the number of layers in each ABP in this sequence, so that $d_{K-1} > n/\Delta$, and $d_k \le 2 d_{k-1}/3$ for $k \in [K]$.  $\mathcal{A}'$ is an ABP with at most $n/\Delta$ layers and size at most $\tau$, which by induction, computes a polynomial of the form
	\[\sum_{i = 1}^n x_i^n + \sum_{j = 1}^r A_j\cdot B_j + R \, ,\]
	where $A_1(\vx), \ldots, A_r(\vx), B_1(\vx), \ldots, B_r(\vx), R(\vx)$ are polynomials such that for every $i$, $A_i(\mathbf{0}) = B_i(\mathbf{0}) = 0$ and $R$ has degree at most $n-1$. Further, the number of error terms, $r$, is at most 
	\[
 		\frac{0.005 n^2} {\Delta} \left( \frac{1}{d_{K-1}} + \frac{1}{d_{K-2}} + \cdots + \frac{1}{d_0} \right).
	\]
	Since $d_k \le \frac{2}{3} \cdot d_{k-1}$, we have that $\frac{1}{d_{k-1}} \le \frac{2}{3} \cdot \frac{1}{d_k}$ for all $k \in [K]$, so that
\[
r \le \frac{0.005n^2}{\Delta} \cdot \frac{1}{1-2/3} \cdot \frac{1}{d_{K-1}} \le \frac{n}{10}
\]
as $d_{K-1} \ge n/\Delta$.	

	At this point, since the formal degree is at most $n$, using \autoref{thm:robustHomog} we get 
	\[
		\tau \geq \abs{\mathcal{A}'} \geq \dfrac{(n/2-r)n}{2\Delta} = \Omega \inparen{\dfrac{n^2}{\Delta}} \, . \qedhere 
	\]
\end{proof}
		
To complete the proof of  \autoref{thm:main ABP}, we now prove \autoref{clm:oneStep}.

\begin{proof}[Proof of \autoref{clm:oneStep}]
	Let $\mathcal{A}_k = \sum_{i=1}^{m} \mathcal{A}_{k,i}$, and for $j \in [d_k]$, let $\ell_{i,j}$ be the number of vertices in layer $j$ of $\mathcal{A}_{k,i}$. Recall that if the number of layers in $\mathcal{A}_{k,i}$ is strictly less than $j$, then we set $\ell_{i,j} = 0$. Let $\ell$ be the total number of vertices in the middle layers of $\mathcal{A}_k$, defined as 
	\[ \ell = \sum_{i=1}^{m} \inparen{\sum_{j \in (d_k/3, 2d_k/3)} \ell_{i,j}} \, .\] 
	Since $\ell \leq \tau \leq \frac{0.001 n^2}{\Delta }$, by averaging, we know that there is a $j_0 \in (d_k/3, 2d_k/3)$, such that 
	\[
		\ell_{j_0} = \sum_{i=1}^{m}\ell_{i,{j_0}} \leq \frac{\ell}{d_k/3} \leq \frac{0.001 n^2}{\Delta } \cdot \frac{1}{d_k/3} \leq 0.005\frac{n^2}{\Delta \cdot d_k} \, .
	\]
	This condition, together with \autoref{lem:shrinkage general}, tells us that there is a multilayered ABP  $\mathcal{A}_{k+1}'$ with at most $2d_k/3$ layers and size at most $\tau$ that computes a polynomial of the form
	\[
		\sum_{i = 1}^n x_i^n + \sum_{j = 1}^r A_j\cdot B_j + R - \sum_{i=1}^{\ell_{j_0}} P_i \cdot Q_i + \delta \, , 
	\]
	where  $P_1, \ldots, P_{\ell}, Q_1, \ldots, Q_{\ell}$ are a set of non-constant polynomials with constant term zero and $\delta \in \F$. Since $\ell_{j_0} \leq 0.005\frac{n^2}{\Delta \cdot d_k}$, the claim follows.	
\end{proof}

\section{Unlayered Algebraic Branching Programs}\label{sec:unlayered}

In this section, we prove \autoref{thm:unlayered ABP}. We begin with the following definition.

\begin{definition}\label{def:cut-ABP}
Let $\mathcal{A}$ be an unlayered ABP over $\F$. Let $s$ and $t$ denote the start and end vertices of $\mathcal{A}$, respectively, and let $v \neq s,t$ be a vertex in $\mathcal{A}$. Denote by $\alpha \in \F$ the constant term of $[s,v]$ and by $\beta \in \F$ the constant term of $[v,t]$.

The \emph{cut} of $\mathcal{A}$ with respect to $v$, denoted $\cut(\mathcal{A},v)$, is the unlayered ABP obtained from $\mathcal{A}$ using the following sequence of operations:
\begin{enumerate}
\item Duplicate the vertex $v$ (along with its incoming and outgoing edges). Let $v_1,v_2$ denote the two copies of $v$.
\item Erase all outgoing edges of $v_1$, and connect $v_1$ to $t$ by a new edge labeled $\beta$.
\item Erase all incoming edges of $v_2$, and connect $s$ to $v_2$ by a new edge labeled $\alpha$. $\hfill\qedhere$
\end{enumerate}
\end{definition}

We now prove some basic properties of the construction in \autoref{def:cut-ABP}.

\begin{claim}\label{claim:cut-properties}
Let $\mathcal{A}$ be an unlayered ABP over $\F$ computing a polynomial $F$, and let $v$ be a vertex in $\mathcal{A}$. Denote $\mathcal{A}' = \cut(\mathcal{A},v)$. Denote by $d$ the depth of $\mathcal{A}$ and by $d_v$ the depth of $v$ in $\mathcal{A}$. Then the following properties hold:
\begin{enumerate}
\item $\mathcal{A}'$ has 1 more vertex and 2 more edges than $\mathcal{A}$.

\item The depth of $\mathcal{A}'$ is at most
\[
\max \set{\depth(\mathcal{A} \setminus \{v\}), d_v + 1, d-d_v+1} \,,
\]
where $\mathcal{A} \setminus \{v\}$ is the ABP obtained from $\mathcal{A}$ by erasing $v$ and all of its adjacent edges.

\item $\mathcal{A}'$ computes a polynomial of the form $F - P \cdot Q - \delta$ where $P$ and $Q$ have no constant term, and $\delta \in \F$.
\end{enumerate}
\end{claim}

\begin{proof}
The first property is immediate from the construction. The second property follows from the following reasoning: each path in $\mathcal{A}'$ is of exactly one of the following types: (a) misses both $v_1$ and $v_2$, (b) passes through $v_1$, or (c) passes through $v_2$. In case (a), the path also appears in the graph of $\mathcal{A} \setminus \{v\}$. In case (b), the only edge going out of $v_1$ is to $t$, and all other edges in the path appear in $\mathcal{A}$, hence the length is at most $d_v+1$. In case (c), the only edge entering $v_2$ is from $s$, hence similarly the path is of length at most $d-d_v+1$.

It remains to show the last property. Let $P' = [s,v]$ and $Q' = [v,t]$ (as computed in $\mathcal{A}$). Denote $P' = P+\alpha$ where $P$ has no constant term and $\alpha \in \F$ and similarly $Q' = Q+\beta$. One may write $F = P' \cdot Q'+R = (P+\alpha)(Q+\beta) + R$ where $R$ is the sum over all paths in $\mathcal{A}$ which do not pass through $v$. In $\mathcal{A'}$, we have that $[s,v_1] =P'$ and $[v_2,t] = Q'$, and thus $\mathcal{A}'$ computes the polynomial
\[
R + \alpha \cdot Q' + P' \cdot \beta = F - P \cdot Q + \alpha \beta. \qedhere
\]
\end{proof}

Our goal is to perform cuts on a strategically chosen set of vertices. In order to select them, will use the following well known lemma of Valiant \cite{Valiant77}, simplifying and improving an earlier result of Erd\H{o}s, Graham and Szemer\'{e}di \cite{EGS75}. For completeness, we also sketch a short proof.

\begin{lemma}[\cite{Valiant77}]\label{lem:Valiant}
Let $G$ be a directed acyclic graph with $m$ edges and depth $d \ge \sqrt{n}$. Then, there exists a set $E'$ of at most $4m / \log n$ edges such that removing $E'$ from $G$ results in a graph of depth at most $d/2$.
\end{lemma}

\begin{proof}
Let $d' \ge d \ge \sqrt{n}$ be a smallest power of 2 larger than $d$, so that $d' \le 2d$. Let $k = \log d'$. 
A \emph{valid labeling} of a directed graph $G = (V,E)$ is a function $f : V \to \set{0,\ldots, N-1}$ such that whenever $(u,v)$ is an edge, $f(u) < f(v)$. Clearly if $G$ had depth $d$ then there is a valid labeling with image $\set{0,\ldots, N-1} = \set{0,\ldots, d-1}$ by labeling each vertex by its depth. Conversely, if there is a valid labeling with image $\set{0,\ldots, N-1}$ then $\depth(G) \le N$.

Let $f$ be a valid labeling of $G$ with image $\set{0,\ldots, d'-1}$ and for $i \in [k]$ let $E_i$ be the set of edges such that the most significant bit in which the binary encoding of the labels of their endpoints differ is $i$. If $E_i$ is removed, we can obtain a valid relabeling of the graph with image $\set{0,\ldots, d'/2-1}$ by removing the $i$-th bit from all labels.

The two smallest sets among the $E_i$-s have size at most $2m/k \le 4m / \log n$ (since $k = \log d' \ge \log n /2$), and removing them gives a valid labeling with image $\set{0,\ldots, d'/4-1}$, and therefore a graph with depth at most $d'/4 \le d/2$.
\end{proof}

We need a slight variation of this lemma, in which we do not pick edges whose endpoints have too small or too large a depth in the graph.

\begin{lemma}\label{lem:graph-depth-reduction}
Let $G$ be a directed acyclic graph with $m$ edges and depth $d \ge \sqrt{n}$. Then, there exist a set $U$ of vertices, of size at most $4m / \log n$, such for every $v \in U$ we have that $d/9 \le \depth(v) \le 8d/9$, and removing $U$ (and the edges touching those vertices) results in a graph of depth at most $3d/4$.
\end{lemma}

\begin{proof}
Let $E$ denote the set of edges of $G$ and $E' \subseteq E$ be the set of edges guaranteed by \autoref{lem:Valiant}. Let $E_1 \subseteq E'$ be the edges in $E'$ whose heads have depth at most $d/9$, and $E_2$ be the edges in $E'$ whose heads have depth at least $8d/9$. Let $E'' = E' \setminus (E_1 \cup E_2)$. Clearly, $|E''| \le |E'| \le 4m/ \log n$. Let $U$ be the set of heads of vertices in $E''$. 

Consider now any path in the graph obtained from $G$ by removing $U$ (and hence in particular $E''$). Given such a path, let $e_1$ be the last edge from $E_1$ in the path which appears before all edges from $E_2$ (if there exists such an edge), and let $e_2$ the first edge from $E_2$ (if any) in the path. We partition the path into three (possibly empty) parts: the first part is all the edges which appear until $e_1$ (including $e_1$); the second part is all the edges after $e_1$ and before $e_2$; the last part consists of all the edges which appear after $e_2$ (including $e_2$). Because the head of $e_1$ is a vertex of depth at most $d/9$, the first part can contribute at most $d/9$ edges. The second part includes only edges from $E \setminus E'$, and thus its length is at most $d/2$. The last part again has depth at most $d/9+1$, as any path leaving a vertex of depth at least $8d/9$ can have at most that many edges (here we add 1 to account for the edge $e_2$ itself, since the assumption is on the depth of the head of $e_2$). Thus, the total length of the path is at most
\[
d/9 + d/2 + d/9 + 1 \le 3d/4\,. \qedhere
\]
\end{proof}

The set of vertices given by the lemma above will be the vertices according to which we will cut the ABP. We describe it in the following lemma, and prove some properties of this operation.

\begin{lemma}\label{lem:unlayered-ABP-depth-reduction}
Let $\mathcal{A}$ be an ABP over a field $\F$ of depth $d \ge \sqrt{n}$ computing a polynomial $F$. Let $\tau$ be the number of vertices and $m$ be the number of edges in $\mathcal{A}$. Then, there exist an unlayered ABP $\mathcal{A}'$, with at most $\tau+ 4m/ \log n$ vertices, at most $m + 8m/\log n$ edges, and depth at most $9d/10$, computing a polynomial of the form $F - \sum_{i=1}^r P_i Q_i - \delta$ where $\delta \in \F$ is a field constant, the $P_i,Q_i$'s have no constant term, and $r \le 4m/\log n$.
\end{lemma}

\begin{proof}
Let $G$ be the underlying graph of the ABP $\mathcal{A}$. Let $U=\set{u_1, \ldots, u_r}$ be the set of vertices guaranteed by \autoref{lem:graph-depth-reduction}, such that $r \le 4m/\log n$. We perform the following sequence of cuts on $\mathcal{A}$. Set $\mathcal{A}_0 := \mathcal{A}$ and for $i \in [r]$, $\mathcal{A}_i = \cut(\mathcal{A}_{i-1}, u_i)$. Finally $\mathcal{A}' = \mathcal{A}_r$.

The statements of the lemma now follow from the properties of cuts as proved in \autoref{claim:cut-properties}. The bound on the number of vertices and edges in $\mathcal{A}'$ is immediate. The claim on the polynomial computed by $\mathcal{A}'$ follows by induction on $i$.

Finally, by induction on $i$, we have that the depth of $\mathcal{A}'$ is at most
\[
\max \{ \depth (\mathcal{A} \setminus U), \depth(u_1) +1 , \ldots, \depth(u_r) + 1,
d-\depth(u_1) + 1, \ldots, d-\depth(u_r)+1 \}\,,
\]
where $\mathcal{A} \setminus U$ is the ABP obtained by removing all vertices in $U$.

By the choice of $U$ as in \autoref{lem:graph-depth-reduction}, for every $i \in [r]$ we have that $d/9 \le  \depth(u_i) \le 8d/9$, and $\depth (\mathcal{A} \setminus U) \le 3d/4$, which implies the required upper bound on the depth of $\mathcal{A}'$ (assuming $n$, and hence $d$, are large enough).
\end{proof}

Repeated applications of \autoref{lem:unlayered-ABP-depth-reduction} give the following statement.

\begin{corollary}\label{cor:depth-reduction-all-the-way}
Let $\mathcal{A}$ be an ABP over a field $\F$, with edge labels of degree at most $\Delta = n^{o(1)}$, computing an $n$-variate polynomial $F$. Further suppose $\mathcal{A}$ has depth at least $\sqrt{n}$, and that the number of edges in $\mathcal{A}$ is at most $ n \log n / (1000 (\log \log n + \log \Delta))$. Let $\tau$ denote the number of vertices in $\mathcal{A}$.

Then, there exists an ABP $\mathcal{A}'$, whose depth is at most $n/\Delta$, which computes a polynomial of the form $F - \sum_{i=1}^r P_i Q_i - \delta$, such that $P_i, Q_i$ are all polynomials without a constant term, $\delta \in \F$ is a field constant, and $r \le n/10$. The number of vertices in $\mathcal{A}'$ is at most $\tau+n/10$.
\end{corollary}

\begin{proof}
Observe that the depth of $\mathcal{A}$ is at most $d := n \log n$. As long as the depth is at least $\sqrt{n}$, apply \autoref{lem:unlayered-ABP-depth-reduction} repeatedly at most $k := 7 (\log \log n + \log \Delta)$ times, to obtain an ABP of depth at most $(0.9)^{k} \cdot d \le n/\Delta$.

The upper bound on the number of summands $P_i Q_i$ and the number of vertices after each application is given as a function of the number of edges, which increases in the process. Hence, we first provide a crude estimate on the number of edges at each step. For $i \in [k]$, let $\mathcal{A}_i$ denote the ABP obtained after the $i$-th application of \autoref{lem:unlayered-ABP-depth-reduction}, and let $m_i$ be the number of edges in that ABP.

We claim that by induction on $i$,  $m_i \le m_0 \cdot (1+8/\log n)^i$. This is true for $i=0$ by definition. For $i \ge 1$, since we maintain the invariant that the depth is at least $\sqrt{n}$, it follows from \autoref{lem:unlayered-ABP-depth-reduction} that
\[
m_i \le m_{i-1} + 8m_{i-1} / \log n = m_{i-1} (1+8/\log n) \le m_0 (1+8/\log n)^{i-1} \cdot (1+8/\log n)\,,
\]
where the last inequality uses the induction hypothesis. Thus, the final ABP has at most
\[
m_{k} \le m_0 (1+8/\log n)^{k} \le 2m_0 =  n \log n / (500 (\log \log n + \log \Delta)) =: M\,
\]
assuming $n$ is large enough (recall that by assumption we have that $\log \Delta = o(\log n)$, so that $\lim_{n \to \infty} (1+8/\log n)^{o(\log n)} = 1$). It is convenient to now use $M$ as a uniform upper bound on the number of edges in all stages of this process, so that each step adds at most $4M /  \log n$ summands and vertices. It now follows that $r$ is at most
\[
\frac{4kM}{\log n} \le \frac{7 (\log \log n + \log \Delta) \cdot 4n}{500( \log \log n + \log \Delta)} \le n/10,
\]
and similarly the total number of vertices added throughout the process is at most $n/10$.
\end{proof}

The lower bound given in \autoref{thm:unlayered ABP} now follows by a simple win-win argument. For convenience, we restate the theorem.

\begin{corollary}\label{cor:unlayered-lower-bound}
Let $\mathcal{A}$ be an ABP over a field $\F$, with edge labels of degree at most $\Delta = n^{o(1)}$, computing $\sum_{i=1}^n x_i ^n$. Then $\mathcal{A}$ has at least $\Omega(n \log n / (\log \log n + \log \Delta)$ edges.
\end{corollary}

\begin{proof}
Let $\tau$ denote the number of vertices in $\mathcal{A}$. If the number of edges is at least $n \log n / (1000( \log \log n + \log \Delta))$, then we already have our lower bound. Else, the number of edges is at most $n \log n / (1000( \log \log n + \log \Delta))$. Now, by \autoref{cor:depth-reduction-all-the-way}, there exists an ABP $\mathcal{A}'$, with $\tau+n/10$ vertices and depth at most $n/\Delta$, computing $\sum_{i=1}^n x_i^n - \sum_{j=1}^r P_i Q_i - \delta$, such that $P_j, Q_j$ have no constant term, $r \le n/10$, and $\delta \in \F$.

It thus follows that $\mathcal{A}'$ has formal degree at most $n$. By \autoref{thm:robustHomog}, it has $\Omega(n^2/\Delta)$ vertices, thus $\tau = \Omega(n^2/\Delta)$, so that the number of edges is also $\Omega(n^2/\Delta)$.
\end{proof}

\section{A Lower bound for formulas}

In this section we prove \autoref{thm:main formula}, which we restate below.

\begin{restatable}{theorem}{mainFormThm}\label{thm:formula lb}
Any algebraic formula computing the polynomial $\esym_{(n, 0.1n)}$ over a field of characteristic greater than $0.1n$ has size at least $\Omega \inparen{n^2}.$
\end{restatable}

Recall that the polynomial we prove the lower bound for is the elementary symmetric polynomial on $n$-variables of degree $0.1n$, which is defined as follows.
\[
\esym(n, 0.1n) = \sum_{S \subseteq [n], \abs{S} = 0.1n}\prod_{j \in S} x_j
\]
A few properties of the elementary symmetric polynomials are needed for our proof, and we first discuss them. These seem to be well known and are not original to this work. 

\subsection{Properties of elementary symmetric polynomials}
\begin{lemma}[\cite{MZ17, LMP19} ]\label{lem:esym homogen case}
	Let $n, d \in \N$ be natural numbers with $2\leq d \leq n$. Then, over any field of characteristic at least $d+1$, the dimension of the variety $V$ defined as 
	\[
		V = \mathbb{V}\inparen{\setdef{\partial_{x_i} \esym_{(n, d)}}{i \in [n]}}
	\]
	is at most $d-2$. Here,  $\partial_{x_i} \esym_{(n, d)}$ is the partial derivative of $\esym_{(n, d)}$ with respect to the variable $x_i$. 
\end{lemma}

We quickly sketch the outline of  a proof of this lemma from \cite{MZ17} for completeness. The following claim  immediately implies the lemma. 

\begin{claim}\label{clm:sing variety of esym}
Let $\va \in \F^n$ be a point in the variety $V$ as defined in \autoref{lem:esym homogen case}. Then, at least $n - (d-2)$ coordinates of $\va$ are equal to $0$. 
\end{claim}

To see that the lemma follows from the claim, observe that \autoref{clm:sing variety of esym} implies that the variety $V$ is a subset of set of $\cup_{S \subseteq n, \abs{S} = d-2} V_S$, where $V_S (\subseteq \F^n)$ is the set of points in $\F^n$ where the coordinates in $S$ are all set to zero, and the other coordinates take all possible values in $\F$. Thus, each $V_S$ is a variety of dimension $d-2$, and $V$ is contained in a union of such varieties. Therefore, the dimension of $V$ cannot exceed $d-2$. We now sketch the proof of the claim. For that we need the following folklore fact, often attributed to Euler. 

\begin{fact}[Euler's formula for differentiation of homogeneous polynomials]
  \label{fact:euler}
Suppose $A(x_1,\ldots, x_k)$ is a homogeneous polynomial of degree $t$. Then $\sum_{i=1}^k x_i \cdot \frac{\partial A}{\partial x_i} = t \cdot A(x_1,\ldots, x_k)$. 
\end{fact}

\begin{proofof}{\autoref{clm:sing variety of esym}}
	The proof is via an induction on the degree $d$ of the elementary symmetric polynomials being considered. Note that we only consider the cases when $2 \leq d \leq n$.
	
	For the base case, we observe that the claim is true whenever $d=2$ and $n \geq 2$. In this case, the partial derivative $\partial_{x_i} \esym_{(n, 1)} = \sum_{j \neq i} x_j$. Now for any $\va \in \F^n$, all these linear forms $\sum_{j \neq i} a_j$ for $i \in [n]$ can vanish if and only if each $a_j$ is zero, thus implying the claim. 
	
	For the induction step, let us fix $d > 2$ and assume that the claim holds for $\esym_{n, r}$ for all $r \in \N$ with $2 \leq r \leq d-1$ and $n \geq r$. We then prove the claim for $\esym_{(n, d)}$ for every $n \geq d$.
	
	For any $n \geq d$, from the definition of elementary symmetric polynomials, observe that 
	\begin{equation}\label{eqn:step1}
	\partial_{x_i} \esym_{(n, d)} = \esym_{(n, d-1)} - x_i \cdot \partial_{x_i} \esym_{(n, d-1)} \, .
	\end{equation}
	If we add up the equations above for each $i \in [n]$, we get 
	\[
		\sum_{i\in [n]} \partial_{x_i} \esym_{(n, d)} = n \cdot \esym_{(n, d-1)} - \sum_{i \in [n]} x_i \cdot \partial_{x_i} \esym_{(n, d-1)} \, .
	\]
	
	\noindent From \autoref{fact:euler}, this can be simplified to 
	\[
		\sum_{i\in [n]} \partial_{x_i} \esym_{(n, d)} = n \cdot \esym_{(n, d-1)} - (d-1) \cdot \esym_{(n, d-1)} \, .
	\]
	Thus, it follows that for every $\va \in \F^n$ where all the first order partial derivatives of $\esym_{(n, d)}$ vanish, it must be the case that $\esym_{(n, d-1)}$ is also zero at $\va$. Moreover, from \autoref{eqn:step1}, this implies that each of the polynomials in the set 
	\[
		\setdef{x_i \cdot \partial_{x_i} \esym_{(n, d-1)}}{i \in [n]} \, 
	\]
	must also vanish at $\va$. Now if $k$ of the coordinates of $\va$ are equal to $0$, and $S$ is the subset of non-zero coordinates of $\va$, it must be the case that the polynomials 
	\[
		\setdef{\partial_{x_i} \esym_{(n-k, d-1)}}{i \in S} \, 
	\]
	vanish at the point $\va_{S} \in \F^{n-k}$, where $\va_S$ is the $n-k$ dimensional vector obtained by projecting $\va$ to the set $S$ of its coordinates.
	
	By the induction hypothesis, we know that for every $r \leq d-1$ and $m \geq r$, any common zero of all the first order partial derivatives of $\esym_{(m, r)}$ must be zero in at least $m-(r-2)$ coordinates. But $\va_S$ is a common zero with support $n-k$ of all the first order partial derivatives of $\esym_{(n-k, d-1)}$.
	
	Thus, if $n-k \geq d-1 (\implies k \leq (n-d+1))$, then by the induction hypothesis, 
	\[
		\text{ no. of zero co-ordinates in } \va_S = 0 \geq n-k-(d-3) \implies k \geq n-d+3 
	\] 
	which would lead to a contradiction. The only other case then, is that 
	\[
		d-1 > n-k \implies k \geq n-(d-2).
	\]
	Thus, $\va$ is zero on at least $n-(d-2)$ of its coordinates. 
\end{proofof}

\noindent We will need the following strengthening of this lemma by Limaye, Mittal and Pareek \cite{LMP19}. 
\begin{lemma}[\cite{MZ17, LMP19}]\label{lem:esym singular locus robust}
	Let $n, d \in \N$ be natural numbers with $2\leq d \leq n$ and let $\F$ be a field of characteristic greater than $d$. Let $R_1, R_2, \ldots, R_n \in \F[\vx]$ be polynomials of degree at most $d-2$. Then, the dimension of the variety $V$ defined as 
	\[
	V = \mathbb{V} \inparen{\setdef{\partial_{x_i} \esym_{(n, d)} - R_i}{i \in [n] }}
	\]
	is at most $d-2$. 
\end{lemma}

\noindent For the sake of completeness, we provide a proof for a slightly more general statement. In this regard, we first introduce a couple of notations.

\begin{notation}
	For a polynomial $f$, let $\Hom_{d}(f)$ denote the $d$-th homogeneous component of $f$. Also, suppose we fix an ordering $<$ on the monomials in $\F[\vecx]$ that is graded according to degree. Then for any ideal $I \subseteq \F[\vecx]$, we denote by $\LM(I)$ the leading monomial ideal of $I$. That is, 
	\[
		\LM(I) = \setdef{m}{m \text{ is a leading monomial (w.r.t. $<$) for some polynomial in } I}. \qedhere
	\]
\end{notation}

\begin{lemma}\label{lem:dimVarWithError}
	Let $\set{f_1, \ldots, f_k}$ be a set of homogeneous polynomials such that $\deg(f_i) = d_i$. Let $R_1, \ldots, R_k$ be polynomials such that $\deg(R_i) < d_i$, and let $I = \inangle{f_1, \ldots, f_k}$ and $I' = \inangle{f_1 + R_1, \ldots. f_k+R_k}$. Then, 
	\[
	\dim(\mathbb{V}(I')) \leq \dim(\mathbb{V}(I))
	\]
\end{lemma}

\begin{proof}
	The proof follows from the following claim.
	
	
	\begin{claim}\label{clm:leadingMon}
		If $\set{f_1, \ldots, f_k}$ are homogeneous polynomials and $\deg(f_i) = d_i$, and let $I,I'$ as in \autoref{lem:dimVarWithError}. Then $\LM(I) \subseteq \LM(I')$.
	\end{claim}
	\noindent Indeed, by the statement of the claim,
	\begin{align*}
	\LM(I) \subseteq \LM(I') &\implies \mathbb{V}(\LM(I')) \subseteq \mathbb{V}(\LM(I)) & (\text{see, e.g., Section 4.2 of } \cite{CLO07})\\
	& \implies \dim(\mathbb{V}(\LM(I'))) \leq \dim(\mathbb{V}(\LM(I)))\\
	& \implies \dim(\mathbb{V}(I')) \leq \dim(\mathbb{V}(I)). 
	\end{align*}
	The last implication follows from the fact that $\dim(\mathbb{V}(\LM(J))) = \dim(\mathbb{V}(J))$ for any ideal $J$ (see, e.g., Section 9.3 of \cite{CLO07}).
\end{proof}

Thus, to finish the proof of \autoref{lem:dimVarWithError}, we need to prove \autoref{clm:leadingMon}.

\begin{proof}[Proof of \autoref{clm:leadingMon}]
	Let $m \in \LM(I)$ be the leading monomial in $\sum_{i=1}^{k} g_i \cdot f_i$, and let $d = \deg(m)$. Since $f_i$ is homogeneous of degree $d_i$, $m$ continues to be the leading monomial in $\sum_{i=1}^{k} \Hom_{d-d_i}(g_i) \cdot f_i$, and hence in $\sum_{i=1}^{k} \Hom_{d-d_i}(g_i) \cdot (f_i + R_i)$. This shows that $m \in \LM(I')$.
\end{proof}

Combining \autoref{lem:esym homogen case} and \autoref{lem:dimVarWithError}, it is easy to see that \autoref{lem:esym singular locus robust} holds.


\subsection{A degree reduction procedure}
We will now prove a statement similar to \autoref{lem:shrinkage general} for the case of formulas. We start with a few preliminary remarks. It is common to define the size of the formula as the number of leaves (which is, up to a factor of 2, the total number of vertices in the formula). For us it is a bit more convenient to define the size of the formulas as the number of leaves labeled by \emph{variables} (rather than field constants), and prove a lower bound on this measure.

The formal degree of a vertex in the formula is defined by induction in the natural way. If $v$ is a leaf labeled by variable, its formal degree is $1$. If $v$ is a leaf labeled by a field constant, its formal degree is $0$. The formal degree of a sum gate is the maximum between the formal degree of its children, and the formal degree of a product gate is the sum of the formal degrees of its children.

The formal degree of the formula is the formal degree of its output gate.

We start by a simple observation which can be proved by induction on the structure of the formula.

\begin{observation}\label{obs:deg to size}
	The size of a formula is at least as large as its formal degree. 
\end{observation}

To avoid cumbersome notation, we identify a formula $\Phi$ with the polynomial it computes. For a vertex $v$ in the formula $\Phi$, we denote by $\Phi_v$ subformula rooted at $v$ (and similarly identify $\Phi_v$ with the polynomial it computes).

\begin{lemma}\label{lem: intermediate degree vertex}
	Let $\Phi \in \F[\vx]$ be a formula of formal degree equal to $d$. Then, for every $t \in \N$ such that $2t \leq d$, there is a vertex $v$ in $\Phi$ of formal degree at least  $t$ and at most $2t-1$. 
	
	Moreover, there are polynomials $h, f \in \F[\vx]$ such that  $\Phi \equiv h\cdot \Phi_v + f$ and for every $\gamma \in \F$, the polynomial $\gamma h + f \in \F[\vx]$ can be computed by a formula of size at most $\abs{\Phi} - \abs{\Phi_v}$.  
\end{lemma}

\begin{proof}
	From the definition of formal degree, recall that for every $+$ gate of formal degree $k$, at least one of its children also has formal degree $k$, and for every $\times$ gate of formal degree $k$, at least one of its children has formal degree at least $\lceil k/2 \rceil$. To get our hands on the vertex $v$ of formal degree in the interval $[t, 2t-1]$ in the formula $\Phi$, we start from the root of $\Phi$ and traverse down, such that in every step if the degree of the current vertex $w$ is at least $2t$, then we traverse down to its child with the larger degree (breaking ties arbitrarily). We stop the first time we reach a vertex $v$ of formal degree at most $2t-1$. The claim is that such a vertex must have formal degree at least $t$.
	
	To see the lower bound of $t$ on the degree of $v$, recall that the formal degree of $\Phi$ is at least $2t$ and by definition, $v$ is the first vertex on the unique path from $v$ to the root of $\Phi$ with formal degree at most $2t-1$. Thus, the parent of $v$ must have formal degree at least $2t$. We also know that at any stage of this walk down from the root, we go from a vertex to its child with the higher formal degree. Therefore, the formal degree of $v$ is at least half of the degree of its parent, i.e., at least $\frac{1}{2} \cdot 2t = t$. 
	
	Consider now that formula obtained from $\Phi$ by replacing all the subtree rooted at $v$ with a leaf labeled by a new variable $y$. This new formula $\Phi'$ computes a polynomial of the form $h \cdot y + f$, and by replacing $y$ with $\Phi_v$ we recover the formula $\Phi$, thus $\Phi \equiv h \cdot \Phi_v + f$.

	The size of $\Phi'$ is $|\Phi| - |\Phi_v| + 1$. For every fixed $\gamma \in \F$ we can actually compute $\gamma h + f$ by replacing the leaf labeled $y$ in $\Phi'$ by the constant $\gamma$ (recall that we do not count leaves labeled by constant in or definition of size).
\end{proof}

We now use \autoref{lem: intermediate degree vertex} inductively to get the following decomposition of a formula. 
\begin{lemma}\label{lem:formula decomposition}
	Let $\Phi$ be a formula of size $s$ and formal degree at most $d$. Then, there exist $k \in \N$,  polynomials $g_1, g_2, \ldots, g_k, h_1, h_2, \ldots, h_k$ and a formula $\Phi'$ such that the following are true. 
	\begin{itemize}
		\item The degree of each $g_i$ is at least $\lfloor d/3 \rfloor$ and at most $2\lfloor d/3 \rfloor -1$ and each $h_i$ has degree at least $1$. 
		\item$\Phi'$ has formal degree at most $2\lfloor d/3 \rfloor$ and size at most $s$.  
		\item The constant term of each $g_i$ and $h_i$ is zero.
		\item There exists a constant $c \in \F$ such that
		\[
		\Phi \equiv \Phi' + \sum_{i = 1}^k g_i h_i  + c \, .
		\] 
		\item $k\lfloor d/3 \rfloor \leq s$. 
	\end{itemize}
\end{lemma}
\begin{proof}
	For brevity, let us assume that $d$ is divisible by $3$, and let $t = d/3$. The proof of the lemma essentially follows from a repeated application of \autoref{lem: intermediate degree vertex}, where we keep extracting vertices of degree at least $t$ till the total degree becomes smaller than $2t$.  We now give the details. 
	
	Since, the formal degree of $\Phi$ is at least $2t$, from \autoref{lem: intermediate degree vertex}, we know that there is a vertex $v_1$ in $\Phi$ with formal degree in the interval  $[t, 2t-1]$ and polynomials $h'_1$ and $f_1$ satisfying the following properties. 
	\begin{itemize}
		\item $\Phi \equiv h'_1\cdot \Phi_{v_1} + f_1 $. 
		\item For every $\gamma \in \F$, $\gamma h'_1 + f_1$ can be computed by a formula of size at most $\abs{\Phi} - \abs{\Phi_{v_1}}$. 
	\end{itemize}
	Let $g'_1 = \Phi_{v_1}$ and $s_1 = |\Phi_{v_1}|$. Let $\alpha, \beta \in \F$ be the constant terms of $g'_1$ and $h'_1$ respectively, and let $g_1$ and $h_1$ be polynomials with constant term zero such that $g'_1 = g_1 + \alpha$ and $h'_1 = h_1 + \beta$. We know that  $\Phi \equiv h'_1\cdot g'_1 + f_1 $, or in other words, $\Phi \equiv (h_1 + \beta)\cdot ({g}_1 + \alpha) + f_1 $. Rearranging the terms, we get 
	\[
	\Phi \equiv g_1h_1 + \beta (g_1 + \alpha) + \alpha(h_1 + \beta) + f_1 - \alpha\beta \, . 
	\]
	This can be re-written as 
	\[
	\Phi \equiv g_1h_1 + \beta g'_1 + \left(\alpha h'_1 + f_1 \right)+ \alpha\beta \, . 
	\]
	Observe that $\alpha h'_1 + f_1$ can be computed by a formula $\Phi_1$ of size at most $s-s_1$. Also, the constant terms of $g_1, h_1$ are both zero and $\beta g'_1$ is computable by a formula $\Phi'_{v_1} \equiv \beta\cdot \Phi_{v_1}$ of formal degree at most $2t-1$ and size $s_1$ (which is also a sub-formula of $\Phi$). In a nutshell, in one application of \autoref{lem: intermediate degree vertex}, we have decomposed $\Phi$ into a sum of a formula $\Phi_{v_1}'$ of size $s_1$ and formal degree at most $2t-1$, a formula $\Phi_1$ of size at most $s-s_1$, a constant term and a product of two constant free polynomials $g_1h_1$, where $g_1$ has degree at least $t$ and $h_1$ has degree at least $1$. To see the lower bound on the degree of $h_1$, note that if $h_1$ is of  degree zero, then it must be identically zero (since it is constant free), and this term just vanishes. 
	
	Now, consider the formula $\Phi_1$. If the formal degree of $\Phi_1$ is at most $2t-1$, then we are already done by letting $\Phi'$ in the statement of the theorem be the sum of $\Phi_1$ and $\Phi'_{v_1}$ (indeed, the size of $\Phi'$ is at most $(s-s_1) + s_1 = s$). Else, observe that the size of $\Phi_1$ is strictly smaller than the size of $\Phi$. All the items of the lemma (except the last item) now follow from a simple induction on the formula size. 
	
	To see the upper bound on $k$ given by the last item, note that the formulas for the polynomials  $g_1', g_2', \ldots, g_k'$ obtained in each of the steps, whose sizes are $s_1, \ldots, s_k$, respectively, are all disjoint subformulas of $\Phi$ and have formal degree at least $t$ each. Thus, by \autoref{obs:deg to size}, $s_i \ge t$ for all $i \in [k]$, and hence, the size of $\Phi$ is at least $kt = k \lfloor d/3\rfloor$. 
\end{proof}

\subsection{A quadratic lower bound for formulas}
We are now ready to prove the lower bound on the formula size computing $\esym_{(n, 0.1n)}$. As was the case with ABPs, before we prove \autoref{thm:formula lb}, we need to prove the lower bound for formulas whose formal degree is at most $0.1n$.

\begin{theorem}\label{thm:robustHomogForm}
	Let $n \in \N$, and let $\F$ be a field of characteristic greater than $0.1n$. Let $A_1(\vx), \ldots, A_r(\vx), B_1(\vx),  \ldots, B_r(\vx)$ and $R(\vx)$ be polynomials such that for every $i$, $A_i(\mathbf{0}) = B_i(\mathbf{0}) = 0$ and $R$ is a polynomial of degree at most $0.1 n-1$. Then, for every $r \leq 0.1n$, any formula over $\F$, of formal degree at most $0.1n$ which computes the polynomial 
	\[\esym_{(n, 0.1n)} + \sum_{j = 1}^r A_j\cdot B_j + R\]
	has at least $0.001 n^2$ vertices.
\end{theorem}

\begin{proof}
	Let $\mathcal{F}$ be a size $s$ formula that computes a polynomial of the form 
	\[\esym_{(n, 0.1n)} + \sum_{j = 1}^r A_j\cdot B_j + R\]
	where $A_1(\vx), \ldots, A_r(\vx), B_1(\vx),  \ldots, B_r(\vx)$ and $R(\vx)$ are as in the statement of the lemma.
	
	By \autoref{lem:formula decomposition}, we know that the polynomial computed by $\mathcal{F}$ has the form
	\[
		R' + \sum_{i=1}^{k} g_i h_i + c
	\]
	where $R'$ has degree at most $2 \floor{0.1n/3} \leq 0.1n -1$, the constant terms of $g_i$ and $h_i$ are zero, and $s \geq k \cdot \floor{0.1n/3}$.
	
	\noindent Thus,
	\[
		\esym_{(n, 0.1n)} + \sum_{j = 1}^r A_j\cdot B_j + R = R' + \sum_{i=1}^{k} g_i h_i + c \implies \esym_{(n, 0.1n)} + R'' = \sum_{i=1}^{k} g_i h_i - \sum_{j = 1}^r A_j\cdot B_j 
	\]
	where $R'' = R - R' -c$ has degree at most $0.1n -1$.
	This shows that if we consider 
	\[
		V = \mathbb{V}\inparen{\setdef{\partial_{x_i} \esym_{(n, 0.1n)} - \partial_{x_i} R''}{i \in [n]}} \qquad V' = \mathbb{V} \inparen{\setdef{g_i, h_i, A_j, B_j}{i \in [k], j \in [r]}},
	\]
	then $V' \subseteq V$ and $V' \neq \emptyset$. Now by \autoref{lem:esym singular locus robust}, we know that $\dim(V) \leq 0.1n-2$ and since $V' \neq \emptyset$, $\dim(V') \geq n - 2k - 2r$ (see, e.g., Section 2.8 of \cite{Smi14}). Hence,
	\[
		n-2k-2r \leq 0.1n - 2 \implies k \geq \frac{n-0.1n-2r+2}{2} = 0.45n - r +1 \geq 0.35n.
	\]
	This implies that $s \geq 0.35n \cdot \floor{\frac{0.1n}{3}} \geq \frac{0.035 n^2}{6} \geq 0.001n^2$. 
\end{proof}

Finally, let us recall the main theorem of this section and complete its proof. 

\mainFormThm*

\begin{proof}
	Let $\mathcal{F}$ be a formula with formal degree $d_0$ which computes the polynomial $\esym_{(n, 0.1n)}$. We assume without loss of generality that the underlying field $\F$ is algebraically closed. If $d_0$ is at most $0.1n$, then by \autoref{thm:robustHomogForm}, we know that $\abs{\mathcal{F}}$ is at least $\Omega(n^2)$ and we are done. Also, if $d_0 > n^2$, then we again have the required lower bound by \autoref{obs:deg to size}. Thus, we may assume that $0.1n \leq d_0 \leq n^2$. 
	
	From this point, the proof is exactly along the same lines as that of \autoref{thm:main ABP}. We iteratively make changes to $\mathcal{F}$, reducing its formal degree geometrically in each step, via \autoref{lem:formula decomposition}, till we get a formula $\mathcal{F}'$ of formal degree at most $0.1n$ that computes a polynomial of the form 
	\[\esym_{(n, 0.1n)} + \sum_{j = 1}^r A_j\cdot B_j + R\]
	where $r \leq 0.1n$ and $A_1(\vx), \ldots, A_r(\vx), B_1(\vx), \ldots, B_r(\vx), R(\vx)$ are polynomials such that for every $i$, $A_i(\mathbf{0}) = B_i(\mathbf{0}) = 0$ and $R$ has degree at most $0.1n - 1$. 
	Once we have this, \autoref{thm:robustHomogForm} gives the required lower bound. As before, in one step we ensure the following.
	
	\begin{claim}\label{clm:oneStepForm}
		Let $\mathcal{F}_k$ be a formula with formal degree $d_k > 0.1 n$ and size at most $\tau$, that computes a polynomial of the form $\esym_{(n,0.1n)} + \sum_{j = 1}^{r_k} A_j\cdot B_j + R$ where $A_1(\vx), \ldots, A_r(\vx), B_1(\vx), \ldots, B_r(\vx), R(\vx)$ are polynomials such that for every $j$, $A_j(\mathbf{0}) = B_j(\mathbf{0}) = 0$ and $R$ has degree at most $0.1n-1$.
		
		If $\tau \leq 0.001 n^2$, then there exists a formula $\mathcal{F}_{k+1}$ with formal degree at most $2 d_k/3$ and size at most $\tau$ which computes a polynomial of the form 
		\[\esym_{(n, 0.1n)} + \sum_{j = 1}^{r_{k+1}} A'_j\cdot B'_j + R'\, ,\]
		such that $r_{k+1} \leq  r_k + \frac{0.0033 n^2}{d_k}$ and $A'_1(\vx), \ldots, A'_{r_{k+1}}(\vx), B'_1(\vx), \ldots, B'_{r_{k+1}}(\vx), R'(\vx)$ are polynomials such that for every $i$, $A'_i(\mathbf{0}) = B'_i(\mathbf{0}) = 0$ and $R'$ has degree at most $0.1n-1$.
	\end{claim}
	
	Before proving \autoref{clm:oneStepForm}, we first complete the proof of \autoref{thm:formula lb} assuming the claim. Set $\mathcal{F}_0 = \mathcal{F}$. Then, $\mathcal{F}_0$ is a formula with formal degree $d_0$ and size at most $\tau$ that computes the polynomial $\esym_{(n, 0.1n)}$. If $\tau \geq 0.001 n^2$, the statement of the theorem follows. 
	
	Otherwise, applying \autoref{clm:oneStepForm} iteratively as long as the formal degree remains $> 0.1n$, we eventually get a formula $\mathcal{F}'$ with formal degree $d' \leq 0.1n$. Let the number of steps taken be $K$, and so let $\mathcal{F}_K = \mathcal{F}'$. Further, let $d_0, \ldots, d_{K-1},d_K$ denote the formal degree of each formula in this sequence. Thus $d_{K-1} > 0.1n$, and $d_k \le 2 d_{k-1}/3$ for every $k \in [K]$. Now $\mathcal{F}'$ has formal degree at most $0.1n$, size at most $\tau$, and computes a polynomial of the form
	\[\esym_{(n, 0.1n)} + \sum_{j = 1}^{r_K} A_j\cdot B_j + R \, ,\]
	where $A_1(\vx), \ldots, A_r(\vx), B_1(\vx), \ldots, B_r(\vx), R(\vx)$ are polynomials such that for every $i$, $A_i(\mathbf{0}) = B_i(\mathbf{0}) = 0$ and $R$ has degree at most $0.1n - 1$. Further, the number of error terms, $r_K$, is at most 
	\[
	0.0033 n^2 \left( \frac{1}{d_{K-1}} + \frac{1}{d_{K-2}} + \cdots + \frac{1}{d_0} \right).
	\]
	Since $d_k \le \frac{2}{3} \cdot d_{k-1}$, we have that $\frac{1}{d_{k-1}} \le \frac{2}{3} \cdot \frac{1}{d_k}$ for all $k \in [K]$. This implies that
	\[
	r_K \le 0.0033 n^2 \cdot \frac{1}{1-2/3} \cdot \frac{1}{d_{K-1}} = \frac{0.0099n^2}{0.1n} = 0.099n \le 0.1n
	\]
	as $d_{K-1} \ge 0.1n$.	
	
	\noindent Finally, since the formal degree is at most $0.1n$, using \autoref{thm:robustHomogForm} we get 
	\[
	\tau \geq \abs{\mathcal{F}'} \geq 0.001 n^2 = \Omega(n^2)\, . \qedhere 
	\]
	We now prove \autoref{clm:oneStepForm}. Again the proof is exactly along the lines of \autoref{clm:oneStep}.
	\begin{proof}[Proof of \autoref{clm:oneStepForm}]
		Let $\mathcal{F}_k$ be a formula with formal degree $d_k > 0.1 n$ and size at most $\tau$, that computes a polynomial of the form $\esym_{(n,0.1n)} + \sum_{j = 1}^{r_k} A_j\cdot B_j + R$. Here $R$ has degree at most $0.1n-1$ and $A_1(\vx), \ldots, A_r(\vx), B_1(\vx), \ldots, B_r(\vx)$ are polynomials such that for every $j$, $A_j(\mathbf{0}) = B_j(\mathbf{0}) = 0$.
		Then, for $\Phi = \mathcal{F}_k$, \autoref{lem:formula decomposition} says that there exists a formula $\mathcal{F}_{k+1} = \Phi'$ of formal degree at most $2d_k/3$ and size at most $\tau$ that computes a polynomial of the form
		\[
		\esym_{(n,0.1n)} + \sum_{j = 1}^{r_k} A_j\cdot B_j + R - \sum_{i=1}^{\ell_k} g_i \cdot h_i + c \, , 
		\]
		where $\ell_k \leq \frac{\tau}{\floor{d_k/3}} \leq \frac{3.3 \times 0.001 n^2}{d_k} = \frac{0.0033 n^2}{d_k}$ for large enough $n$. Further, $g_1, \ldots, g_{\ell_k}, h_1, \ldots, h_{\ell_k}$ are a set of non-constant polynomials with constant term zero and $c \in \F$. Thus, if we set $R' = R+c$, $A'_i = A_i, B'_i = B_i$ for every $i \in [r_k]$ and $A'_{r_k+j} = g_j, B'_{r_k+j} = h_j$ for every $j \in [\ell_k]$, the claim follows.	
	\end{proof}
\end{proof}

\section{Open problems}
We conclude with some open problems. 
\begin{itemize}
\item A natural open question here is to prove an improved lower bound for unlayered algebraic branching programs. In particular, in the absence of an obvious  non-trivial upper bound,  it seems reasonable to conjecture that any unlayered ABP computing the polynomial $\sum_{i = 1}^n x_i^n$  has size at least $\Omega(n^{2-o(1)})$. 
\item Yet another question which is natural in the context of this work and remains open is to prove stronger lower bounds for ABPs. As a first step towards this, the question of proving super-quadratic lower bound for homogeneous algebraic formulas might be more approachable.
 
\end{itemize}

\section*{Acknowledgements}
We are thankful to Ramprasad Saptharishi for helpful discussions at various stages of this work and to Vishwas Bhargava for asking us whether the ABP lower bound in \autoref{thm:main ABP} also gives an $\Omega(n^2)$ lower bound for algebraic formula for any multilinear polynomial, and for many subsequent discussions.

We are also thankful to Nutan Limaye for telling us about the results in \cite{LMP19} and letting us use \autoref{lem:esym singular locus robust} in our proof of \autoref{thm:main formula}.  

\bibliographystyle{customurlbst/alphaurlpp}
\bibliography{masterbib/references,masterbib/crossref}
\end{document}